\documentclass[a4paper,11pt]{article}

\usepackage{arxiv}

\usepackage[utf8]{inputenc} 
\usepackage[T1]{fontenc}    
\usepackage{hyperref}       
\usepackage{url}            
\usepackage{booktabs}       
\usepackage{amsfonts}       
\usepackage{nicefrac}       
\usepackage{microtype}      
\usepackage{lipsum}
\usepackage{amsmath}
\usepackage{amsthm}
\usepackage{amssymb}
\usepackage{mathtools}
\newtheorem{theorem}{Theorem}
\newtheorem{lemma}[theorem]{Lemma}
\newtheorem{corollary}[theorem]{Corollary}

\newtheorem{conj}[theorem]{Conjecture}
\theoremstyle{definition}
\newtheorem{defn}[theorem]{Definition}
\theoremstyle{definition}
\newtheorem{remark}[theorem]{Remark}
\numberwithin{theorem}{subsection}
\usepackage{graphicx}
\usepackage{subcaption}
\usepackage{natbib}
\usepackage{tikz}
\usepackage{tikz-network}
\usepackage{tikz-cd}

\newcommand{\ba}{\begin{eqnarray}}
\newcommand{\ea}{\end{eqnarray}}
\newcommand{\nn}{\nonumber}

\title{Embedding and approximation theorems for echo state networks}

\author{
  Allen Hart  \\
  Department of Mathematical Sciences\\
  University of Bath\\
  Bath BA2 7AY, UK \\
  \texttt{a.hart@bath.ac.uk} \\
   \And
 James Hook \\
    Department of Mathematical Sciences\\
  University of Bath\\
  Bath BA2 7AY, UK \\
  \texttt{j.l.hook@bath.ac.uk} \\   
  \And
 Jonathan Dawes \\
    Department of Mathematical Sciences\\
  University of Bath\\
  Bath BA2 7AY, UK \\
  \texttt{j.h.p.dawes@bath.ac.uk} \\  
}

\begin{document}
\maketitle

\begin{abstract}
Echo State Networks (ESNs) are a class of single-layer recurrent neural networks that have enjoyed recent attention. In this paper we prove that a suitable ESN, trained on a series of measurements of an invertible dynamical system, induces a $C^1$ map from the dynamical system's phase space to the ESN's reservoir space. We call this the Echo State Map. We then prove that the Echo State Map is generically an embedding with positive probability.

Under additional mild assumptions, we further conjecture that the Echo State Map is almost surely an embedding. For sufficiently large, and specially structured, but still randomly generated ESNs, we prove that there exists a linear readout layer that allows the ESN to predict the next observation of a dynamical system arbitrarily well. Consequently, if the dynamical system under observation is structurally stable then the trained ESN will exhibit dynamics that are topologically conjugate to the future behaviour of the observed dynamical system.

Our theoretical results connect the theory of ESNs to the delay-embedding literature for dynamical systems, and are supported by numerical evidence from simulations of the traditional Lorenz equations. The simulations confirm that, from a one dimensional observation function, an ESN can accurately infer a range of geometric and topological features of the dynamics such as the eigenvalues of equilibrium points, Lyapunov exponents and homology groups.
\end{abstract}

\subsection*{Keywords}
Reservoir computing; liquid state machine; time series analysis; Lorenz equations; dynamical system; delay embedding; Persistent Homology; recurrent neural networks.

\newpage

\section{Introduction}

An Echo State Network (ESN) is a single-layer recurrent neural network composed of a trainable readout layer connected to a reservoir of randomly initialized, and randomly coupled, untrainable `neurons'. This architecture has been investigated and used by many authors since the seminal papers by \cite{Jaeger2001} and \cite{doi:10.1162/089976602760407955}. \cite{TANAKA2019100} present a review of ESNs, among other recurrent neural network models, under the umbrella term \emph{reservoir computing}.

The wide range of problems to which the ESN framework has been applied include speech recognition \citep{SKOWRONSKI2007414}, learning grammatical structure \citep{TONG2007424}, and financial time series prediction \citep{TimeSeries}, \citep{LIN20097313}. Several authors including \cite{BiologicalESN} have also discussed how the ESN is a plausible model for the information processing performed by biological neurons. Most ambitiously, \cite{10.1007/978-3-540-25940-4_14} discuss ESNs in the context of \emph{building by 2050, a team of fully autonomous
humanoid robots to beat the human winning team of the FIFA Soccer World Cup}.

The ESN has associated to it a \emph{reservoir state} denoted $r_k \in \mathbb{R}^n$ at time $k$. The structure of the recurrent layer is described by an $n\times n$ matrix $A$ that is the weighted adjacency matrix of the system of $n$ `neurons'. If neuron $i$ is not connected to neuron $j$ then $A_{ij} = 0$, and if they are connected with some weight $a \in \mathbb{R}$ then $A_{ij} = a$. Connections need not be symmetric, so in general $A_{ij} \neq A_{ji}$. Typically, $A$ is sparse and has approximately $1\%$ of its entries non-zero. Connection weights are usually i.i.d. random variables, and typically are chosen to be either uniformly distributed on a fixed interval, or Gaussian. The ESN also contains an $(n \times m)$ input matrix $W^{\text{in}}$, where $m$ is the dimension of the training data. Like the reservoir $A$, $W^{\text{in}}$ is also populated with i.i.d random variables. Finally, the ESN has an activation function $\varphi:\mathbb{R}^n \to \mathbb{R}^n$, for which there are several standard choices, for example $\text{tanh}$ (performed component-wise). 

The operation of the ESN is divided into two phases: an initial training phase, followed by an autonomous phase.
During the training phase, the ESN is trained on a given input time series denoted by vectors $u_0$, $u_1$, $u_2$ ... $u_K$ each in $\mathbb{R}^m$. We will assume in this paper that the input sequence is bounded, though we note the recent work of \cite{JMLR:v20:19-150} establishes a framework that encompasses unbounded input sequences as well. We will also assume in this paper that $m = 1$, so that we consider a scalar input time series. The \emph{reservoir state} at time $k$ is defined by choosing an initial state e.g. $r_1 = (0 , 0 , ... , 0)^{\top}$ and defining subsequent states recursively by
\begin{align}
    r_{k+1} = \varphi( Ar_k + W^{\text{in}}u_{k}). \nn
\end{align}
Having computed the the new reservoir states $r_1, r_2 ... r_K$, the output matrix $W^{\text{out}}$ is fitted to solve
the optimisation problem
\begin{align}
    \min_{W^{\text{out}}}\sum_{k=1}^{K} \lVert W^{\text{out}}r_{k} - a_{k} \rVert^2 + \lambda\lVert W^{\text{out}} \rVert^2_2, \nn
\end{align}
where $a_k$ is some known target sequence we want the ESN to mimic, often taken to be equal to the input sequence $u_k$, and $\lambda>0$ is a regularisation parameter. Minimisation problems of this kind are often referred to as ridge regression, or Tikhonov, or $L^2$ regularisation. Having trained the output matrix $W^{\text{out}}$ the reservoir states $s_{k}$ can then be liberated from their reliance on the driving input $u_k$ and evolve under the autonomous dynamical system defined by
\begin{align}
    &v_{k+1} = W^{\text{out}}s_k, \nonumber \\
    &s_{k+1} = \varphi( As_k + W^{\text{in}}v_{k+1}), \nonumber
\end{align}
where $s_0 = r_K$. If the training has been successful, then the trained ESN should provide good predictions of the future time series $v_1 \approx u_{K+1}, v_2 \approx u_{K+2}$, etc, and future evolution of the reservoir state $s_1 \approx r_{K+1}, s_2 \approx r_{K+2}$ etc. The viewpoint we take here clearly distinguishes between the training phase of the ESN where it is an externally-driven dynamical system, and the `test' phase where we consider it as an autonomous dynamical system in $\mathbb{R}^n$. 

In complete generality the process defining $u_k$ could be anything, including a realisation of a random process. However, importantly, throughout this paper, we will restrict our attention and assume that $u_0,u_1,u_2,\ldots$ are one dimensional observations of an invertible discrete-time dynamical system with evolution operator $\phi \in \text{Diff}^1(M)$ observed via a function $\omega \in C^1(M,\mathbb{R})$ on a compact manifold $M$. In particular $u_0 = \omega(x), u_1 = \omega \circ \phi(x), u_2 = \omega \circ \phi^{2}(x), u_3 = \omega \circ \phi^{3}(x)$, etc. The model we have in mind is that $\phi$ is the evolution operator for a time $\Delta t$ of a set of Lipschitz ordinary differential equations on $M$. Illustrations of the training and autonomous phases are shown in Figure \ref{ESN_fig}.

\begin{figure}
  \centering
  \begin{subfigure}{\textwidth}
  \caption{Training phase}
\begin{tikzpicture}
\pgfmathsetseed{1} 
\node[inner sep=0pt] (lorenz) at (-10,0)
    {\includegraphics[width=.4\textwidth]{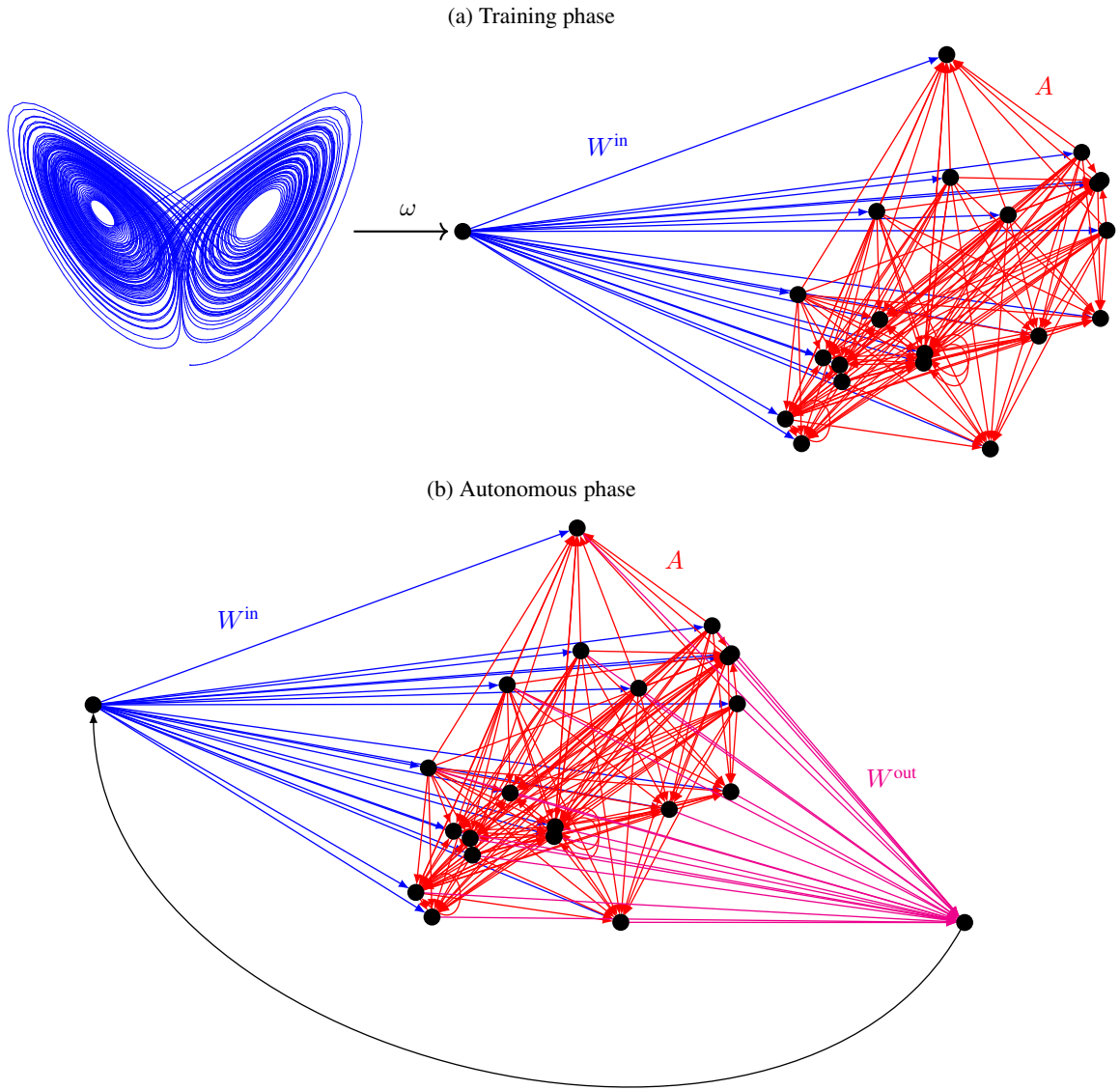}};
\Vertex[x=-6, y=0, style={color=black},size=.1]{0}
\foreach \x in {1,2,...,20}{
\Vertex[x=rand*3,y=rand*3,style={color=black},size=.1]{\x}}
\foreach \x in {1,2,...,20}{
\Edge[color = blue,Direct,lw=.5pt](0)(\x)}
\foreach \x in {1,2,...,11}{
\foreach \y in {9,10,...,20}{
\Edge[color = red,Direct,lw=.5pt](\x)(\y)}
}
\draw[->,thick, color=black] (-7.5,0) -- (-6.2,0);
\node[text=blue](Win) at (-4,1.2){$W^{\text{in}}$};
\node(omega) at (-6.75,0.3){$\omega$};
\node[text=red](A) at (2,2){$A$};
\end{tikzpicture}
    \label{Driven_ESN_fig}
\end{subfigure}
    
    \begin{subfigure}{\textwidth}
  \centering
  \caption{Autonomous phase}
\begin{tikzpicture}
\pgfmathsetseed{1} 
\Vertex[x=-6, y=0, style={color=black},size=.1]{0}
\foreach \x in {1,2,...,20}{
\Vertex[x=rand*3,y=rand*3,style={color=black},size=.1]{\x}}
\foreach \x in {1,2,...,20}{
\Edge[color = blue,Direct,lw=.5pt](0)(\x)}
\foreach \x in {1,2,...,11}{
\foreach \y in {9,10,...,20}{
\Edge[color = red,Direct,lw=.5pt](\x)(\y)}
}
\Vertex[x=6, y=-3, style={color=black},size=.1]{21}
\foreach \x in {1,2,...,20}{
\Edge[color = magenta,Direct,lw=.5pt](\x)(21)}
\Edge[color = black,Direct,lw=.5pt,bend=75](21)(0)
\node[text=blue](Win) at (-4,1.2){$W^{\text{in}}$};
\node[text=red](A) at (2,2){$A$};
\node[text=magenta](Wout) at (5,-1){$W^{\text{out}}$};
\end{tikzpicture}
    \label{Autonomous_ESN_fig}
\end{subfigure}
    \caption{(a) During the training phase the ESN observes a dynamical system via the function $\omega \in C^1(M,\mathbb{R})$; this sequence of observations is distributed into the nodes in the reservoir $r$ by the linear map $W^\text{in}$. (b) After training, in the autonomous phase, the driving is replaced by the output created by the best-fit linear map $W^{\text{out}}$. These images were produced using the TikZ-network package developed by \cite{TikZ_network_manual}.} 
\label{ESN_fig}
\end{figure}

The idea to draw training data from a dynamical system was by \cite{Jaeger78} who drew observations from a trajectory through the Mackey-Glass attractor. We were attracted to the idea by a recent paper by \cite{Pathak2017} who trained ESNs on the Lorenz equations and the Kuramoto--Sivashinsky equation (in one spatial dimension). In particular, we conjecture that under the right technical conditions an ESN with random reservoir matrix and input matrix trained on a one dimensional observation of a dynamical system will embed the system into the reservoir space almost surely. We call this the ESN Embedding Conjecture (Conjecture \ref{ESN_embedding_conjecture}). We believe this conjecture is true as a consequence of \cite{TakensThm} theorem stating that a generic delay observation map is an embedding. This connection between Takens' delay embedding theorem and the ESN was remarked on by \cite{Jaeger2001} and has been discussed in several later works including by \cite{1556081}, \cite{416607}, \cite{4118282}, \cite{5645205}, \cite{Lokse2017}, \cite{YEO2019671}, and \cite{Vlac_2019}. We go on to prove that our statement of the ESN Embedding Conjecture holds with probability $\alpha > 0$. 
We finally prove that when the ESN does successfully embed a structurally stable dynamical system into its reservoir, there exists a trainable readout layer such that the autonomous phase of the ESN will adopt the topology of the driving dynamical system. 
We call this the ESN Approximation Theorem (Theorem \ref{ESN_approximation_theorem}). This theorem complements the results of \cite{GRIGORYEVA2018495} and \cite{Gonon2020} stating that the ESN (with tunable and randomly initialised $A$ and $b$ respectively) is a universal approximator of discrete-time fading memory filters.

To demonstrate the theory we present numerical evidence that an ESN trained on a numerically integrated trajectory of the Lorenz system can replicate several of the Lorenz system's geometric and topological properties. In particular, we computed the Lyapunov exponents of the ESN autonomous phase and compared them to the known exponents of the Lorenz system. We also compared the eigenvalues of the system linearisation on the Lorenz system's fixed points to the eigenvalues of the linearisation on the fixed points belonging to the ESN autonomous phase. Finally we compared the homology of the driven and autonomous reservoir attractors to the Lorenz attractor using persistent homology. For the reader unfamiliar with persistent homology \cite{Ghrist2007} offers an excellent primer.

The remainder of the paper is set out as follows. In section 2 we present basic definitions and define a family of maps that captures the effect on the reservoir state of training with increasing amounts of data. In section~\ref{sec:esm} we prove that the family converges to a $C^1$ map that we call the Echo State Map. We conjecture in section~\ref{sec:embedding} that generically the Echo State Map is an embedding. In section~\ref{sec:approx} we prove an ESN Approximation Theorem that guarantees
that the autonomous dynamics of the ESN is (in a suitable sense) conjugate via a diffeomorphism to the original dynamical system on which the ESN was trained. in section~\ref{sec::numerical} we present numerical results supporting the theory. 

\section{Theory of ESNs}

Our analysis makes use of several different norms. In particular, if $x \in \mathbb{R}^m$ is a vector then $\lVert x \rVert$ is the Euclidean norm, and for $A$ a matrix then $\lVert A \rVert_2$ is the matrix $2$ norm. If $f$ is a real valued function, then $\lVert f \rVert_{\infty}$ will denote the supremum norm and if $f$ is continuously differentiable then we will use
the $C^1$ norm $\lVert f \rVert_{C^1}$ defined by
\begin{align}
    \lVert f \rVert_{C^1} := \lVert f \rVert_{\infty} + \lVert Df \rVert_{\infty} \nn 
\end{align}
where $D$ is the derivative operator.

\subsection{The Echo State Network}
\label{sec:theory}
We begin our summary of the background to Echo State Networks (ESNs) with a definition.

\begin{defn} (Echo State Network)
    Let the activation function $\sigma$ be a function $\sigma \in C^1(\mathbb{R},(-1,1))$ that has its derivative take values in the range $(0,1)$.
    Let $n \in \mathbb{N}$, $A$ be a real $n \times n$ matrix, and $W^{\text{in}}$ a real $n \times 1$ matrix. Let $b_i \in \mathbb{R} \ \forall i \in \{ 1 , ... , n \}. $ Let $I_n := [-1,1]^n$ and define the function $\varphi : \mathbb{R}^n \to I_n$ component-wise by
    \begin{align}
        \varphi_i(r) = \sigma(r_i + b_i) \ \forall i \in \{1 , ... , n \}. \label{eqn:phi_i} 
    \end{align}
    We then define an Echo State Network (ESN) of size $n$ to be the triple $(\varphi, A , W^{\text{in}})$. 
\end{defn}
$\sigma$ is often chosen to be the hyperbolic function tanh, though other choices of activation function abound in the machine learning literature. The conditions on these functions are sometimes less restrictive than those imposed above on $\sigma$; other common choices of activation function include the linear unit (also known as the identity map) and the rectified linear unit (often abbreviated relu) defined by
\begin{align}
    \text{relu}(r_i) = 
    \begin{cases}
        r_i &\text{ if } r_i > 0 \\
        0 &\text{ otherwise. } \nn
    \end{cases}
\end{align}
\cite{pmlr-v15-glorot11a} discuss how Recurrent Neural Networks supported by a relu activation function are less prone to the \emph{vanishing gradient problem} than sigmoidal activation functions.
More exotic activation functions include radial basis functions, which take the shape of bell curves. Throughout this paper however, we will restrict ourselves to activation functions as defined above, i.e. functions $\sigma \in C^1(\mathbb{R},(-1,1))$ who's derivatives take values in $(0,1)$.

\subsection{The Echo State Map}
\label{sec:esm}

We will begin by introducing a family of functions that describe the mapping between this time series of observations and the reservoir state; this will be of fundamental importance throughout the remainder of the paper.

\begin{defn}
    (Echo State Family)
   Let $M$ be a compact $m$-manifold and $n \in \mathbb{N}$. Let $A$ be an $n \times n$, and let $W^{\text{in}}$ an $n \times 1$ matrix: let the triple $(\varphi,A,W^{\text{in}})$ be an ESN. Let the discrete dynamical system be $\phi \in \text{Diff}^1(M)$ and let the observation function $\omega \in C^1(M,\mathbb{R})$. Let the family of functions $F = \{f^{r_0}_k : M \to I_n : r_0 \in I_n, \ k \in \mathbb{N}_0 \}$ be defined as follows:
\begin{align}
    f^{r_0}_0(x) &= r_0 \nonumber \\
    f^{r_0}_{k+1}(x) &= \varphi(A f^{r_0}_k \circ \phi^{-1}(x) + W^{\text{in}}\omega(x) ). \nn
\end{align}
We call the set of functions $F$ the \emph{Echo State Family}. \label{def:ESF}
\end{defn}

To provide some intuition as to where this family came from, we observe that $f^{r_0}_{k}$ is the function that takes a point $x \in M$ and first applies the inverse evolution operator $k$ times, yielding the past state $\phi^{-k}(x)$ of the dynamical system. A list of $k+1$ observations $ \omega \circ \phi^{-k}(x) , \ \omega \circ \phi^{1-k}(x) , \ \omega \circ \phi^{2-k}(x) , \ ... $ are then obtained, in sequence, forward from this point. An ESN with initial reservoir state $r_0$ is trained on this list of inputs, and its reservoir state is then exactly given by the value of $f^{r_0}_{k}(x)$. The function $f^{r_0}_{k}$ is therefore the map induced by $k+1$ steps of the training phase of the ESN, i.e. it sends a point $x \in M$ to reservoir space $I_n$ according to its one dimensional history. 
Our plan for the upcoming section is to show that for any $r_0 \in I_n$
\begin{align}
    f^{r_0} := \lim_{k \to \infty} f^{r_0}_{k}
    \nonumber
\end{align}
exists, and that $f^{r_0} = f^{s_0} =: f$ for any $r_0,s_0 \in I_n$. We will call $f$ the Echo State Map, and show it is continuously differentiable, i.e. $f \in C^1(M,\mathbb{R}^n)$. These results will appear together and called the Echo State Mapping Theorem. Equivalently, we could say the Echo State Map $f$ is
the unique $C^1$ generalized synchronisation (in the sense described by \cite{PhysRevLett.76.1816} ) between a pair of unidirectionally coupled systems, the dynamics given by $\phi$ and the driven ESN phase. 

We will further conjecture that $f$ is a $C^1$ embedding almost surely, and therefore (almost surely) it is a topology-preserving map from the manifold $M$ to the reservoir space $I_n$. We will call this the ESN Embedding Conjecture, and go on to prove a partial result that $f$ is a $C^1$ embedding with positive probability.

\begin{theorem}
        \label{c1_theorem}
        (Echo State Mapping Theorem) With the notation and hypotheses of Definition~\ref{def:ESF}, and the further assumption that $\lVert A \rVert_2 < \min(1,1 / \lVert D\phi^{-1} \rVert_{\infty})$,
        there exists a unique solution $f \in C^1(M,\mathbb{R}^n)$ of the equation
        \begin{align}
            f = \varphi (A f \circ \phi^{-1} + W^{\text{in}}\omega ) \nn
        \end{align}
        \label{recursion_relation}
        such that for all $r_0 \in I_n$ the sequence $f^{r_0}_k$ converges  in the $C^1$ topology to $f$ as $k \to \infty$. We call $f$ the Echo State Map.
    \end{theorem}
    
    \begin{proof}
        Let $\tilde\Psi:C^1(M,\mathbb{R}^n) \to C^1(M,\mathbb{R}^n)$ be defined by
        \begin{align}
            \tilde\Psi(f) = \varphi (A f \circ \phi^{-1} + W^{\text{in}}\omega) \nn
        \end{align}
        then we can see that 
        \begin{align}
            f^{r_0}_k = \tilde\Psi(f^{r_0}_{k-1})
            = \tilde\Psi^k(f^{r_0}_0). \nn
        \end{align}
        Now, we will show that $\tilde\Psi$ is a contraction mapping and therefore has a unique fixed point $f \in C^1(M,\mathbb{R}^n)$ by the contraction mapping theorem \citep{Banach1922}. This will complete the proof. 
        \ba
            \lVert \tilde\Psi(f) - \tilde\Psi(g) \rVert_{C^1} & = & \lVert \varphi(Af \circ \phi^{-1} + W^{\text{in}}\omega) - \varphi(Ag \circ \phi^{-1} + W^{\text{in}}\omega) \rVert_{C^1} \nn \\
            & \leq & \lVert Af \circ \phi^{-1} + W^{\text{in}}\omega - Af\circ\phi^{-1} - W^{\text{in}}\omega \rVert_{C^1} \text{ because $\varphi$ is contracting in $C^1$} \nn \\
            & = & \lVert A(f \circ \phi^{-1} - g \circ \phi^{-1}) \rVert_{C^1} \nn \\
            & \leq & \lVert A \rVert_2 \lVert f \circ \phi^{-1} - g \circ \phi^{-1} \rVert_{C^1} \nn \\
            & = & \lVert A \rVert_2( \lVert f \circ \phi^{-1} - g \circ \phi^{-1} \rVert_{\infty} + \lVert Df \circ \phi^{-1} D \phi^{-1} - Dg \circ \phi^{-1} D\phi^{-1} \rVert_{\infty}  ) \nn \\
            & \leq & \lVert A \rVert_2( \lVert f \circ \phi^{-1} - g \circ \phi^{-1} \rVert_{\infty} + \lVert D\phi^{-1} \rVert_{\infty} \lVert Df \circ \phi^{-1}  - Dg \circ \phi^{-1}  \rVert_{\infty}  ) \nn \\
            & \leq & \lVert A \rVert_2 \max(1,\lVert D\phi^{-1} \rVert_{\infty}) \lVert f - g \rVert_{C^1} \nn
        \ea 
        and $\lVert A \rVert_2 \max(1,\lVert D\phi^{-1} \rVert_{\infty}) < 1$, so we have that $\tilde{\Psi}$ is contracting.
    \end{proof}
     
We remark here that if $\phi$ is obtained by the discretisation of a continuous time flow with a small time step, the evolution operator $\phi$ is close to the identity map, so $\lVert D\phi^{-1} \rVert_{\infty}$ is close to 1. Consequently, the condition 
$\lVert A \rVert_2 < \min(1,1 / \lVert D\phi^{-1} \rVert_{\infty})$
is not much more restrictive than enforcing $\lVert A \rVert_2 < 1$.

\subsection{The ESN Embedding Theorem}
     \label{sec:embedding}

    In this section we will discuss the conditions under which the Echo State Map $f \in C^1(M,\mathbb{R}^n)$ is a $C^1$ embedding (i.e. an injective immersion whose domain and image are diffeomorphic). We will also conjecture that for a generic observation function $\omega$ and random matrices $A$ and $W^{\text{in}}$, the Echo State Map $f$ is a $C^1$ embedding almost surely. To set the scene for these results, we first recall Whitney's Weak Embedding Theorem and Takens' Theorem for delay observation maps.

    \begin{theorem}
    (Whitney's Weak Embedding Theorem) Let $M$ be a compact $m$-manifold and choose $n \in \mathbb{N}$ such that $n > 2m$. Then the set of $C^r$ embeddings is generic in $C^r(M,\mathbb{R}^n)$ with respect to the Whitney $C^1$ topology (This is the topology on $C^1(M,\mathbb{R}^n)$ induced by the $C^1$-norm). 
    \label{WWET}
\end{theorem}
\begin{proof}
    \cite{whitney_embedding_thm}.
\end{proof}
    \begin{corollary}
        Let $M$ be a compact $m$-manifold and $n \in \mathbb{N}$ such that $n > 2m$. Let $A$ be an $n \times n$ matrix for which $\lVert A \rVert_2 < \text{min}(1/\lVert D\phi^{-1} \rVert_{\infty}, 1)$. Let $W^{\text{in}}$ be an $n \times 1$ matrix, and let the triple $(\varphi,A,W^{\text{in}})$ be an ESN. As usual, let $\phi \in \text{Diff}^1(M)$, and $\omega \in C^1(M,\mathbb{R})$. If $n > 2m$, then the ESM $f \in C^1(M,\mathbb{R}^n)$ is a limit point in the Whitney $C^1$ topology of the set of $C^1$ embeddings.
        \label{WeakESN}
    \end{corollary}

    \begin{proof}
        $f \in C^1(M,\mathbb{R}^n)$ by Theorem \ref{c1_theorem} so, by the Weak Whitney Embedding Theorem (Theorem~\ref{WWET}), $f$ is a limit point of the $C^1$ embeddings with respect to the Whitney $C^1$ topology.
    \end{proof}
    From Corollary \ref{WeakESN} it is clear that the Echo State Map $f$ is always close to an embedding, but this says nothing about necessary or sufficient conditions for $f$ to actually \emph{be} an embedding. In fact $f$ may never actually be an embedding. That said, since embeddings are generic in the space $C^1(M,\mathbb{R}^n)$ we expect heuristically that a function in $C^1(M,\mathbb{R}^n)$ that is assembled without explicitly desiring that it is not an embedding, is overwhelmingly likely actually to be an embedding. This suggests (heuristically) that a generic Echo State Map $f$ is indeed an embedding. The first step we take toward proving this is to introduce Takens' Theorem.

    \begin{theorem}
    (Huke's Formulation of Takens' Theorem) Let $M$ be a compact manifold of dimension $m$. Suppose $\phi \in \text{Diff}^2(M)$ has the following two properties:
    \begin{itemize}
    \item[(1)] $\phi$ has only finitely many periodic points with periods
    less than or equal to $2m$.
    \item[(2)] If $x \in M$ is any periodic point with period $k < 2m$ then the eigenvalues of the derivative $D \phi^k$ at $x$ are distinct.
    \end{itemize}
    Then for a generic $C^2$ observation function $\omega \in C^2(M,\mathbb{R})$ the $(2m+1)$ delay observation map $\Phi_{(\phi,\omega)}:M \to \mathbb{R}^{2m+1}$ defined by
\begin{align}
    \Phi_{(\phi,\omega)}(x) := ( \omega(x) , \omega \circ \phi (x) , \omega \circ \phi^2 (x) , ... , \omega \circ \phi^{2m}(x) ) \nn 
\end{align}

is a $C^1$ embedding.
\end{theorem}
\begin{proof}
    \cite{Hukes_thm}. 
\end{proof}
Huke's proof that $\Phi_{(\phi,\omega)}$ is a $C^1$ embedding for generic $\omega$ comprises two steps. First, he shows that $\Phi_{(\phi,\omega)}$ is a $C^1$ embedding for an open subset of $C^2$ observation functions, and second, he shows that $\Phi_{(\phi,\omega)}$ is an embedding for a dense subset of all $C^2$ observation functions. The first step (to prove openness) is fairly simple while the second (the proof of density) is long and delicate. A brief summary of the density part of the proof is as follows. An arbitrary $C^2$ observation function $\omega$ is carefully perturbed on each open set in a cover of the manifold $M$ such that $\omega$ becomes immersive on each set. The observation function $\omega$ is then perturbed again on each open set in the cover in order to make $\omega$ injective, with care taken to ensure $\omega$ remains immersive on each open set. This procedure is applied separately to open sets which contain periodic points and open sets that do not.
We believe it is possible to build on this result and modify the proof of Huke's Theorem in order to prove an ESN Embedding Conjecture in the form that we now state. 

    \begin{conj}
        (ESN Embedding Conjecture) Let $M$ be a compact $m$-manifold and $n \in \mathbb{N}$ such that $n>2m$. Let $A$ be an $n \times n$ matrix with $\lVert A \rVert_2 < \text{min}(1 / \lVert D\phi^{-1} \rVert_{\infty} , 1)$, and $W^{\text{in}}$ a $n \times 1$ matrix, and let the triple $(\varphi,A,W^{\text{in}})$ be an ESN. Let $\omega \in C^2(M,\mathbb{R})$ and let $\phi \in \text{Diff}^2(M)$ (and possibly requiring additional properties), and let $A,W^{\text{in}}$ be generic matrices in the topology induced by the matrix $2$-norm. Then the Echo State Map $f \in C^1(M,\mathbb{R}^n)$ is a $C^1$ embedding.
        \label{ESN_embedding_conjecture}
    \end{conj}

We now summarise our partial success towards proving this conjecture. In particular we can establish the properties analogous to the first part of Huke's proof of Takens' Theorem: we will show that 
the set of triples $(A,W^{\text{in}},\omega)$ of reservoir matrix, input matrix, and observation function for which $f$ is a $C^1$ embedding, is open and non-empty. Consequently, for a generic observation function $\omega$, and matrices $A$ and $W^{\text{in}}$ drawn from a distribution with full support (if the pdf is well defined, it is greater than 0 over its domain), $f$ is a $C^1$ embedding with probability $\alpha > 0$. To prove the full ESN Embedding Conjecture, all that remains is to show that the triples $(A,W^{\text{in}},\omega)$ for which $f$ is an embedding are dense in the space of admissible triples, but this is no easy task, so we will be satisfied here with the proof of only openness and non-emptiness. 

\begin{lemma}
Let $M$ be a compact $m$-manifold and $n \in \mathbb{N}$. Let $A$ be an $n \times n$ matrix, and suppose that $\lVert A \rVert_2 < \text{min}(1/\lVert D\phi^{-1} \rVert_{\infty}, 1)$. As usual
let $W^{\text{in}}$ a $n \times 1$ matrix, let the triple $(\varphi,A,W^{\text{in}})$ be an ESN, $\phi \in \text{Diff}^1(M)$ and $\omega \in C^1(M,\mathbb{R})$. Define the set
$\Omega := \{ (A,W^{\text{in}},\omega) \mid f_{A,W^{\text{in}},\omega} \text{ is a } C^1 {\text{ embedding.}}\}$. Then the set $\Omega$ is open in the $C^1$ topology.
\label{open_lemma}
\end{lemma}

\begin{proof}
    First we define the map $\Psi$ that associates the ESM $f$ to the triple $(A,W^{\text{in}},\omega)$; let $\Psi : (A,W^{\text{in}},\omega) \to C^1(M,\mathbb{R}^n)$ be defined by
        $\Psi(A,W^{\text{in}},\omega) = f_{A,W^{\text{in}},\omega}$.
We now argue as follows. Since $C^1$ embeddings form an open subset of $C^1(M,\mathbb{R})$, and the inverse image of a continuous map is open, it suffices to show that $\Psi$ is continuous in order to
then conclude that $\Omega$ is open. To show continuity of $\Psi$ we must prove that if $(A_n, W^{\text{in}}_n,\omega_n)_{n \in \mathbb{N}} \to (A,W^{\text{in}},\omega)$ then
    $\lVert \Psi(A_n , W^{\text{in}}_n , \omega_n) - \Psi(A,W^{\text{in}},\omega) \rVert_{C^1} \to 0$. 

To lighten the notation we will write $f$ for $f_{A, W^{\text{in}}, \omega}$ and $f_n$ for $f_{A_n , W^{\text{in}}_n, \omega_n}$. As a preliminary result we estimate as follows:
\begin{align}
        \lVert A_n f_n \circ \phi^{-1} - A f \circ \phi^{-1} \rVert_{C^1} 
        & =  \lVert A_n f_n \circ \phi^{-1} - A f \circ \phi^{-1} \rVert_{\infty} \nn \\
        & +  \lVert A_n Df_n \circ \phi^{-1} D\phi^{-1} - A Df \circ \phi^{-1} D\phi^{-1} \rVert_{\infty} \label{An_fn_phi}  \\
        &  \qquad \text{ by definition of the $C^1$ norm } \nn \\
        & \leq  \lVert A_n f_n \circ \phi^{-1} - A f \circ \phi^{-1} \rVert_{\infty}
        + \lVert D\phi^{-1} \rVert_{\infty} \lVert A_n D f_n \circ \phi^{-1} - A D f \circ \phi^{-1} \rVert_{\infty} \nn \\
        & \leq  \lVert A_n f_n - A f \rVert_{\infty}
        + \lVert D \phi^{-1} \rVert_{\infty} \lVert A_n Df_n - A Df \rVert_{\infty} \nn \\
        & \leq  \text{max}(1,\lVert D\phi^{-1} \lVert_{\infty})( \lVert A_n f_n - A f \rVert_{\infty} + \lVert A_n Df_n - A Df \rVert_{\infty} ) \nn  \\
        & \leq  \text{max}(1,\lVert D\phi^{-1} \rVert_{\infty}) \lVert A_n f_n - Af \rVert_{C^1} \nn \\
        &= \tau \lVert A_n f_n - Af \rVert_{C^1} \label{A_f_phi} 
\end{align}
    where we have defined $\tau := \text{max}(1,\lVert D\phi^{-1} \rVert_{\infty})$. We will prove one more preliminary result: that $\lVert f_n \rVert_{C^1}$ is bounded. We can see that $\lVert f_n \rVert_{\infty}$ is bounded by boundedness of $\varphi$ so all that remains is to bound $\lVert Df_n \rVert_{\infty}$. Since
    \begin{align}
        f_n &= \varphi( A_n f_n \circ \phi^{-1} + W^{\text{in}}_n \omega_n ) \nn
    \end{align}
we compute directly that
\begin{align}
        Df_n &= D\varphi( A_n f_n \circ \phi^{-1} W^{\text{in}}_n \omega_n )( A_n D f_n \circ \phi^{-1} D \phi^{-1} + W^{\text{in}}_n D \omega_n ) \nn 
\end{align}
from which we can estimate that
\ba
\lVert Df_n \rVert_{\infty} & = & \lVert D\varphi( A_n f_n \circ \phi^{-1} W^{\text{in}}_n \omega_n )( A_n D f_n \circ \phi^{-1} D \phi^{-1} + W^{\text{in}}_n D \omega_n ) \rVert_{\infty} \nn \\
        & \leq & \rVert A_n D f_n \circ \phi^{-1} D \phi^{-1} + W^{\text{in}}_n D \omega_n \rVert_{\infty} \nn \\
        & \leq & \lVert A \rVert_{2} \lVert D f_n \circ \phi^{-1} \rVert_{\infty} \lVert D \phi^{-1} \rVert_{\infty} + \lVert W^{\text{in}}_n \omega_n \rVert_{\infty} \nn \\
        & < & \bar{\rho} \lVert D f_n \circ \phi^{-1}  \rVert_{\infty} \lVert D \phi^{-1} \rVert_{\infty} + \lVert W^{\text{in}}_n D \omega_n \rVert_{\infty} \text{ \ where $\bar{\rho} = \sup_{n \in \mathbb{N}}\lVert A_n \rVert_{2} < 1$} \nn \\
        & = & \bar{\rho} \lVert D f_n \rVert_{\infty} \lVert D \phi^{-1} \rVert_{\infty} + \lVert W^{\text{in}}_n D \omega_n \rVert_{\infty} \nn \\
        & < & \bar{\rho} \lVert D f_n \rVert_{\infty} \lVert D \phi^{-1} \rVert_{\infty} + \nu \nn
    \ea
    where $\nu$ is a bound for the sequence $\lVert W^{\text{in}}_n D \omega_n \rVert_{\infty}$, which we know exists because $\lVert W^{\text{in}}_n D \omega_n \rVert_{\infty}$ converges. Now upon rearrangement 
    \begin{align}
        \lVert D f_n \rVert_{\infty} < \frac{\nu}{1 - \bar{\rho} \lVert D \phi^{-1} \rVert_{\infty}}, \nn 
    \end{align}
    hence we have bounded $\lVert Df_n \rVert_{\infty}$ and $\lVert f_n \rVert_{\infty}$ thus we have a bound for $\lVert f_n \rVert_{C^1}$, which we will call $\mu$.
    Now, for all $\epsilon > 0$ there exists $n \in \mathbb{N}$ such that both
    \begin{align}
        \lVert A_n - A \rVert_2 < \frac{\epsilon(1 - \tau \lVert A \rVert_2)}{2\tau\mu}
        \label{An_A}
    \end{align}
    and 
    \begin{align}
        \lVert W^{\text{in}}_n\omega_n - W^{\text{in}}\omega \rVert_{C^1} < \frac{\epsilon(1 - \tau \lVert A \rVert_2)}{2}.
        \label{Wn_W}
    \end{align}
    Armed with these estimates we can now compute that
    \ba
\lVert f_n - f \rVert_{C^1}
        & = & \lVert \varphi(A_n f_n \circ \phi^{-1} + W^{\text{in}}_n \omega_n) - \varphi(A f \circ \phi^{-1} + W^{\text{in}} \omega) \rVert_{C^1} \text{ by Theorem~\ref{recursion_relation} } \nn \\ 
        & \leq & \lVert A_n f_n \circ \phi^{-1} + W^{\text{in}}_n \omega_n - A f \circ \phi^{-1} - W^{\text{in}} \omega \rVert_{C^1} \text{ because $\varphi$ is contracting} \nn \\
        & \leq & \lVert A_n f_n \circ \phi^{-1} - A f \circ \phi^{-1} + W^{\text{in}}_n \omega_n - W^{\text{in}} \omega \rVert_{C^1} \nn \\
        & \leq & \lVert A_n f_n \circ \phi^{-1} - A f \circ \phi^{-1} \rVert_{C^1} + \lVert W^{\text{in}}_n \omega_n - W^{\text{in}} \omega \rVert_{C^1} \nn \\
        & \leq & \tau \lVert A_n f_n - A f \rVert_{C^1} + \lVert W^{\text{in}}_n \omega_n - W^{\text{in}} \omega \rVert_{C^1} \text{ by equations \eqref{An_fn_phi}-\eqref{A_f_phi}} \nn \\
        & \leq & \tau \lVert A_n f_n - Af_n + Af_n - Af \rVert_{C^1} + \lVert W^{\text{in}}_n \omega_n - W^{\text{in}} \omega \rVert_{C^1} \nn \\
        & \leq & \tau( \lVert A f_n - A f \rVert_{C^1} + \lVert A_n f_n - Af_n \rVert_{C^1} ) + \lVert W^{\text{in}}_n \omega_n - W^{\text{in}} \omega \rVert_{C^1} \nn \\
        & \leq & \tau \lVert A \rVert_2 \lVert f_n - f \rVert_{C^1} + \tau \lVert f_n \rVert_{C^1} \lVert A_n - A \rVert_{2} + \lVert W^{\text{in}}_n \omega_n - W^{\text{in}} \omega \rVert_{C^1} \nn \\
        & < & \tau \lVert A \rVert_2 \lVert f_n - f \rVert_{C^1} + \tau\mu \lVert A_n - A \rVert_{2} + \lVert W^{\text{in}}_n \omega_n - W^{\text{in}} \omega \rVert_{C^1} \nn \\
        & < & \tau \lVert A \rVert_2 \lVert f_n - f \rVert_{C^1} + \frac{\epsilon(1 - \tau \lVert A \rVert_2)}{2} + \frac{\epsilon(1 - \tau \lVert A \rVert_2)}{2}
        \text{ by equations \eqref{An_A} and \eqref{Wn_W}} \nn \\
        & < & \tau \lVert A \rVert_2 \lVert f_n - f \rVert_{C^1} + \epsilon(1 - \tau \lVert A \rVert_2). \nn 
    \ea
    
    Hence, rearranging we see that
    $\lVert f_n - f \rVert_{C^1}(1 - \tau \lVert A \rVert_2) < \epsilon(1 - \tau \lVert A \rVert_2)$ which implies
    $\lVert f_n - f \rVert_{C^1} < \epsilon$ as required.
    \end{proof}

To prove non-emptiness we construct an explicit reservoir matrix $A$ and input matrix $W^{\text{in}}$ for which the Echo State Map $f$ is an embedding, using a trick borrowed from \cite{4118282}.

First, for a given observation function $\omega$ we define $\Lambda_\omega$ to be the subset of matrices $A$ and $W^\text{in}$ for which the associated map $f$ is a $C^1$ embedding:
\ba
\Lambda_{\omega}:= \{ (A,W^{\text{in}}) \mid f_{A,W^{\text{in}},\omega} \text{ is a } C^1 {\text{ embedding.}}\} \label{def:lambdaomega}
\ea 

\begin{lemma}
    Let $M$ be a compact $m$-manifold and $n \in \mathbb{N}$. Let $A$ be an $n \times n$ matrix
    and suppose $\lVert A \rVert_2 < \text{min}(1 / \lVert D\phi^{-1} \rVert_{\infty} , 1)$. Let $W^{\text{in}}$ be an $n \times 1$ matrix and let the triple $(\varphi,A,W^{\text{in}})$ be an ESN.
    Suppose that $\phi \in \text{Diff}^2(M)$ has the following two properties:
    \begin{itemize}
    \item[(1)] $\phi$ has only finitely many periodic points with periods
    less than or equal to $2m$.
    \item[(2)] If $x \in M$ is any periodic point with period $k < 2m$ then the eigenvalues of the derivative $D \phi^k$ at $x$ are distinct.
    \end{itemize}
    Then for a generic $\omega \in C^2(M,\mathbb{R})$, $\Lambda_{\omega}$ is non-empty.
\label{non_empty_lemma}
\end{lemma}

\begin{proof}
    Let
    \begin{align}
        A = \frac{1}{2}        
        \begin{bmatrix}
    0 & 0 & 0 & \dots & 0 \\
    1 & 0 & 0 & \dots & 0 \\
    0 & 1 & 0 & \dots & 0 \\
    0 & 0 & 1 & \dots & 0 \\
    0 & 0 & 0 & \ddots & 0 
        \end{bmatrix} \nn
    \end{align}
    and $W^{\text{in}}_1 = 1, \ W^{\text{in}}_j = 0$ for $2 \leq j \leq n$. Then the ESM
    \begin{align}
        f := 
        \begin{bmatrix}
            \varphi_1 \circ \omega \\
            \varphi_2 \circ 2^{-1} \varphi_1 \circ \omega \circ \phi^{-1} \\
            \varphi_3 \circ 2^{-1} \varphi_2 \circ 2^{-1} \varphi_1 \circ \omega \circ \phi^{-2} \\
            \vdots \\
            \varphi_n \circ 2^{1-n} \varphi_{n-1} \dots 2^{-1} \varphi_1 \circ \omega \circ \phi^{-n+1} 
        \end{bmatrix}, \nn
    \end{align}
    where $\varphi_i(r_i) = \sigma(r_i + b_i)$ is the $i$th component function of $\varphi$, as defined in~\eqref{eqn:phi_i},
        solves the equation
        \begin{align}
            f = \varphi(A f \circ \phi^{-1} + W^{\text{in}}\omega). \nn
        \end{align}
        We can see moreover that $f \equiv g \circ \Phi_{(\phi,\omega)}$ where 
        \begin{align}
            g := 
        \begin{bmatrix}
            \varphi_1 \\
            \varphi_2 \circ 2^{-1} \varphi_1  \\
            \varphi_3 \circ 2^{-2} \varphi_2 \circ 2^{-1} \varphi_1 \\
            \vdots \\
            \varphi_n \circ 2^{1-n} \varphi_{n-1} \dots 2^{-1} \varphi_1
        \end{bmatrix} \nn
        \end{align}
        and $\Phi_{(\phi,\omega)}$ is the delay observation map 
        \begin{align}
            \Phi_{(\phi,\omega)}(x) = ( \omega(x) , \omega \circ \phi^{-1} (x) , \omega \circ \phi^{-2} (x) , ... , \omega \circ \phi^{-n+1}(x) ). \nn
        \end{align}
        By design, each $\varphi_i$ is a $C^1$ embedding hence $g$ is a $C^1$ embedding. For generic $\omega \in C^2(M,\mathbb{R})$ the delay observation map $\Phi_{(\phi,\omega)}$ is also a $C^1$ embedding, thanks to Takens' Theorem. Noting that the composition of $C^1$ embeddings is a $C^1$ embedding completes the proof.
\end{proof}

\begin{theorem}
    (Weak ESN Embedding Theorem) Let $M$ be a compact $m$-manifold and $n \geq 2m+1$. Let $A$ be a random variable with a distribution that has full support on the space of $n \times n$ matrices for which $\lVert A \rVert_2 < \text{min}(1 / \lVert D\phi^{-1} \rVert_{\infty}, 1)$, and let $W^{\text{in}}$ be a random variable with a distribution that has full support on the space of $n \times 1$ matrices, and let the triple $(\varphi,A,W^{\text{in}})$ be an ESN. Suppose $\phi \in \text{Diff}^2(M)$ has the following two properties:
    \begin{itemize}
    \item[(1)] $\phi$ has only finitely many periodic points with periods
    less than or equal to $2m$.
    \item[(2)] If $x \in M$ is any periodic point with period $k < 2m$ then the eigenvalues of the derivative $D \phi^k$ at $x$ are distinct.
    \end{itemize}
    Then for a generic observation function $\omega \in C^2(M,\mathbb{R})$ the Echo State Map $f$ is a $C^1$ embedding with probability $\alpha > 0$.
\end{theorem}

\begin{proof}
    The space of all observation functions $\omega \in C^2(M,\mathbb{R})$ such that the delay observation map $\Phi_{(\phi,\omega)}$ is an embedding is generic in $C^2(M,\mathbb{R})$, thanks to Takens' Theorem. For any one of these observation functions, the set
        $\Lambda_{\omega}$ defined in~\eqref{def:lambdaomega} 
    is non-empty by Lemma~\ref{non_empty_lemma} and open by Lemma~\ref{open_lemma}. Since $A$, and $W^{\text{in}}$ are random variables with full support, they take values in $\Lambda$ with probability $\alpha > 0$.
\end{proof}

\begin{remark}
The Embedding Conjecture and Weak ESN Embedding Theorem state that under the right conditions $f$ is an embedding. In practical examples we cannot compute $f$ exactly because it is obtained in the limit of infinitely many past observations. In practice, if we have $k$ observations the best we can do is to use all available observations and compute $f^{r_0}_{k}$. Fortunately, the set of $C^1$ embeddings is open in the $C^1$ topology, and $f^{r_0}_{k}$ converges to $f$ in this topology, so there exists a sufficiently large number $\ell$ of previous observations such that for all $k > \ell$, $f^{r_0}_{k}$ is an embedding.
\end{remark}


The ESN Embedding Conjecture also admits a biological interpretation. Consider an organism with a (primitive) nervous system (`brain') comprised of neurons. Neurons are connected to each other with random connection weights (including zero) representing the strength of the connection (or no connection). The adjacency matrix forms the reservoir matrix $A$. The reservoir state $r$ is a vector representing the firing rate of every neuron. Suppose that the organism has a sensory organ connected to the brain which at any point in time senses a scalar measure of the environment, for example an average environmental light intensity. The connection weight from the sensory organ to the $i$th neuron is then the $i$th entry of $W^{\text{in}}$. Suppose that the light intensity depends on the state of the environment which evolves as a high dimensional dynamical system. Then the nervous system and sensory organ together operate as an ESN. Since the entries of $A$ and $W^{\text{in}}$ are random variables, the ESN Embedding Conjecture states that the dynamics of the environment are indeed embedded into the nervous system without the nervous system needing to possess any special structure provided by learning or natural selection. The embedding of the natural world into the brain is obtained `for free'. This leaves cognition, defined as `the art of performing computation on our representation of the environment', as the faculty that requires optimisation by natural selection or learning.

\subsection{The ESN Approximation Theorem}
\label{sec:approx}

In this section we will state and prove the ESN Approximation Theorem - that an ESN which successfully embeds a dynamical system into the reservoir space can approximate the system's dynamics during the autonomous phase, hence replicate the topology of a structurally stable dynamical system. We will use several preliminary results introduced over the proceeding subsections.

\subsubsection{The Universal Approximation Theorem}

The first major result we will use to prove the ESN Approximation Theorem is the Universal Approximation Theorem. This theorem is highly celebrated in the literature on mathematical analysis of neural networks, and states that smooth functions and any number of their derivatives can be approximated by single layer neural network with sufficiently many neurons. In this section we recall this theorem and then present an extension suitable for ESNs, to take account of the fact that for an ESN the neural network weights $v_i$ and biases $b_i$ are randomly chosen but then fixed; only the output weights $w_i$ can be chosen to give a good approximation to an input function $f$. We will use the Universal Approximation Theorem presented by \cite{HORNIK1990551}, because it concerns smooth functions \emph{and any number of their derivatives} while the earlier seminal paper by \cite{Cybenko1989} does not.

\begin{defn}
    ($\ell$-finite) Let $\ell \in \mathbb{N}_0$. Then we say an $\ell$-times differentiable scalar function $\sigma \in C^{\ell}(\mathbb{R})$ is $\ell$-finite if
    \begin{align}
        0 < \int_{\mathbb{R}} \bigg| \frac{d^{\ell} \sigma}{d x^{\ell}} \bigg| dx < \infty. \nn 
    \end{align}
\end{defn}

\begin{remark}
    The activation function $\sigma \in C^1(\mathbb{R},(-1,1))$ with derivative in the range $(0,1)$ is $1$-finite; meaning $\ell$-finite with $\ell = 1$.
\end{remark}

\begin{theorem}
    (Universal Approximation Theorem) If the activation function $\sigma$ is $\ell$-finite, then for all $0 \leq m \leq \ell$  functions $g : I_n \to \mathbb{R}$ of the form
    \begin{align}
        g(x) = \sum_{j = 1}^N w_j \sigma ( v_j^{\top}x + b_j) \nn 
    \end{align}
    are dense in $C^{m}(I_n,\mathbb{R})$.
\end{theorem}

\begin{proof}
    \cite{HORNIK1990551}. 
\end{proof}
    
    The Universal Approximation Theorem essentially states that if we are interested in approximating a function $f$ to some tolerance $\epsilon$ we can create a neural network of size $N$ and modify the weights until the network approximates $f$ to the tolerance $\epsilon$. We want to slightly extend the theorem for our purposes.
    Recall that an ESN has random reservoir weights comprising the matrix $A$ and random input weights comprising the matrix $W^{\text{in}}$, and it is only the output connection weights $W^{\text{out}}$ that are trained. 
    We therefore want to show that for any continuously differentiable function $f$ and a sufficiently large neural network with random weights $v_i$ and biases $b_i$, we can choose linear readout weights $w_i$ such that the resulting neural network approximates $f$ arbitrarily well with probability arbitrarily close to 1. We will call this the Random Universal Approximation Theorem (RUAT), and remark that the RUAT is highly related to Theorem 2.1 appearing in the seminal paper on Extreme Learning Machines by \cite{HUANG2006489}. We can also view the RUAT as a special case of Theorem 1 presented by \cite{Gonon2020}, who prove a stronger result in the more general context of filters.
    
    The idea behind the proof of the RUAT is as follows. First we note that there is a neural network $\hat{g}$ that approximates $f$ by the Universal Approximation Theorem. Then, we create sample sequences of weights and biases $v_i,b_i$ by repeated draws from appropriate random variables. There will eventually be some randomly generated samples $v_j,b_j$ that are close to each of the weights and biases of the network $\hat{g}$. From this list of weights and biases in the sample sequences we select those that match closely, and so create a neural network $g$, choosing linear readout weights $w_i$ either to match the respective weight in $\hat{g}$ or choosing to set $w_i = 0$ in order effectively to discard those samples $v_i, b_i$ that not close to values in $\hat{g}$. Now by construction $g$ is a good approximation to $\hat{g}$ which is itself a good approximation to $f$. The details are presented in the following lemma and theorem.
    
    \begin{lemma}
\label{RUAT_lemma}
    Let $(X_j)_{j \in \mathbb{N}}$ be a sequence of i.i.d. random variables and $S_1, \ldots, S_{\ell}$ be a list of $\ell$ events, and suppose that for each $i$ (and for any $j$ since they are i.i.d.) there exists $\theta_i$ such that $\mathbb{P}(X_j \in S_i) = \theta_i > 0$. Then for all $\alpha \in (0,1)$ there exists $N \in \mathbb{N}$ such that
    \begin{align}
        \mathbb{P}\big( \exists \text{ injective } \phi : \{ 1, \ldots, \ell \} \to \{ 1, \ldots, N \} : X_{\phi(i)} \in S_i, \ \ \forall \ i \in \{ 1, \ldots, \ell \} \big) > \alpha. \nn 
    \end{align}
\end{lemma}

\begin{proof}
    First, fix $\alpha \in (0,1)$.
    Then define the set $\{n_0, \ldots, n_\ell\}$ as follows.
    Set $n_0 = 0$ and for any $i \in \{ 1 , ... , \ell \}$ let
    \begin{align}
        n_i - n_{i-1} := \text{ceil}\bigg( \frac{\log(1 - \alpha^{1 / \ell})}{\log(1-\theta_i)} \bigg) + 1. \nn 
    \end{align}
    Finally, set $N=n_{\ell}$. Then we can calculate that
    
\ba 
\mathbb{P} \big( \exists \text{ injective } \phi 
: X_{\phi(i)} \in S_i \ \forall i \in \{1, \ldots, \ell \} \big) 
& > &
\mathbb{P} \big(\forall \ i \in \{ 1, \ldots, \ell \} \ \exists \ j \in \{1 + n_{i-1}, \ldots, n_i \} : X_j \in S_i \big) \nn \\ 
& = &
        \prod_{i=1}^{\ell} \mathbb{P}\big( \exists j \in \{1 + n_{i-1}, \ldots, n_i \} : X_j \in S_i \big) \nn \\ 
        & = &
        \prod_{i=1}^{\ell} 1 - \mathbb{P}\big( X_j \notin S_i \ \forall j \in \{1 + n_{i-1}, \ldots, n_i \} \big) \nn \\ 
        & \geq & 
        \prod_{i=1}^{\ell} 1 - (1 - \theta_i)^{n_i - n_{i-1}} \nn \\ 
        & = &
        \prod_{i=1}^{\ell} 1 - (1 - \theta_i)^{\text{ceil}\big(\log(1 - \alpha^{1 / \ell}) / \log(1-\theta_i) \big) + 1} \nn \\ 
        & > &
        \prod_{i=1}^{\ell} 1 - (1 - \theta_i)^{\big(\log(1 - \alpha^{1 / \ell}) / \log(1-\theta_i) \big)} \nn \\ 
        & = & 
        \prod_{i=1}^{\ell} 1 - \exp\bigg(\frac{\log(1 - \alpha^{1 / \ell})}{\log(1-\theta_i)}\log(1 - \theta_i)\bigg) \nn \\ 
        & = & \prod_{i=1}^{\ell} 1 - (1 - \alpha^{1/\ell}) = \prod_{i=1}^{\ell} \alpha^{1/\ell} = \alpha. \nn
    \ea
    
\end{proof}

\begin{theorem} (Random Universal Approximation Theorem)
    Let $I_n$ denote the unit hypercube of dimension $n$ and let $f \in C^1(I_n , \mathbb{R})$.
    Let $\sigma \in C^1(\mathbb{R})$ be $1$-finite, and let $(b_j)_{j \in \mathbb{N}}$, $(v_j)_{j \in \mathbb{N}}$ be sequences of i.i.d. random variables with full support.  Then for any $\alpha \in (0,1)$ and $\epsilon > 0$ there exists some natural number $N \in \mathbb{N}$ such with, probability greater than $\alpha$, there exist real numbers $w_1, \ldots, w_N \in \mathbb{R}$ such that the \emph{random neural network} $g : I_n \to \mathbb{R}$ defined by
    \begin{align}
        g(x) = \sum_{j=1}^{N} w_j \sigma (v_j^{\top} x + b_j) \nn 
    \end{align}
    satisfies
    \begin{align}
    \lVert f - g \rVert_{C^1} < \epsilon. \nn
    \end{align}
    \label{blehh}
\end{theorem}
\begin{proof}
    First, by the Universal Approximation Theorem we know that for any $\epsilon > 0$ there exists a neural network $\hat{g} : I_n \to \mathbb{R}$ of size $\ell$ defined by
\begin{align}
        \hat{g} (x) = \sum_{i = 1}^\ell \hat{w}_i \sigma(\hat{v}_i^{\top} x + \hat{b}_i) \nn 
\end{align}
    such that
    \begin{align}
    \lVert f - \hat{g} \rVert_{C^1} < \frac{\epsilon}{2}. \label{eqn:fghat} 
    \end{align}
Now, consider two sequences of i.i.d. random variables $(b_j)_{j \in \mathbb{N}}$ and $(v_j)_{j \in \mathbb{N}}$ with full support, and let $X_j := (b_j , v_j)$. Fix $\epsilon>0$  and define a collection of $\ell$ events $S_1 , ... , S_{\ell}$ by
\begin{align}
    S_i := \bigg\{ (b,v) \in \mathbb{R} \times \mathbb{R}^n : \lVert \sigma(\hat{v}_i^{\top} \cdot + \hat{b}_i) - \sigma(v^{\top} \cdot + b) \rVert_{C^1} < \frac{\epsilon}{2 \ell \max_k(\hat{w}_k)} \bigg\}, \nn 
\end{align}
where the weights $\hat{w}_k$ are given by the form of the network $\hat{g}$.
Observe that each of the $S_i$ have strictly positive measure, so there exists
$\theta_i>0$ such that $\mathbb{P}(X_j \in S_i) > \theta_i > 0 \ \forall j \in \mathbb{N}$. 
Hence it follows by Lemma~\ref{RUAT_lemma} that for all $\alpha \in (0,1)$ there exists $N \in \mathbb{N}$ such that
    \begin{align}
        \mathbb{P}\big( \exists \text{ injective } \phi : \{ 1, \ldots, \ell \} \to \{ 1, \ldots, N \} : X_{\phi(i)} \in S_i \ \forall i \in \{1, \ldots, \ell \} \big) > \alpha. \nn 
    \end{align}
Now, on the event
\begin{align}
    \exists \text{ injective } \phi : \{ 1, \ldots, \ell \} \to \{ 1, \ldots, N \} : X_{\phi(i)} \in S_i \ \forall i \in \{1, \ldots, \ell \} \nn 
\end{align}
we define
\begin{align}
    w_j := 
    \begin{cases}
        \hat{w}_i \text{ if } \phi(i) = j \\
        0 \text{ otherwise }
    \end{cases} \nn 
\end{align}
for all $j \in \{ 1, \ldots, N \}$, and define the \emph{random neural network} $g:I_n \to \mathbb{R}$ by
\begin{align}
    g(x) = \sum^{N}_{j=1}w_j \sigma( v_j^{\top} x + b_j ). \nn
\end{align}
Now observe
\ba
\lVert \hat{g} - g \rVert_{C^1}
    & = & \bigg\lVert \sum_{i = 1}^\ell \hat{w}_i \sigma(\hat{v}_i^{\top} \cdot + \hat{b}_i) - \sum_{j = 1}^{N}w_j\sigma(v^{\top}_j \cdot + b_j) \bigg\rVert_{C^1} \nn \\
    & = & \bigg\lVert \sum_{i = 1}^\ell \hat{w}_i \big( \sigma ( \hat{v}_i^{\top} \cdot + \hat{b}_i ) - \sigma \big( v_{\phi(i)}^{\top} \cdot + b_{\phi(i)} \big)  \big) \bigg\rVert_{C^1} \nn \\ 
    & \leq & \sum_{i = 1}^\ell \hat{w}_i \bigg\lVert \big( \sigma ( \hat{v}_i^{\top} \cdot + \hat{b}_i ) - \sigma \big( v_{\phi(i)}^{\top} \cdot + b_{\phi(i)} \big)  \big) \bigg\rVert_{C^1} \nn \\
    & < & \sum_{i = 1}^\ell \frac{\hat{w}_i \epsilon}{2 \ell \max_k(\hat{w}_k)} < \frac{\epsilon}{2}. \nn
\ea
Combining this with~\eqref{eqn:fghat} and using the triangle inequality we obtain
\ba 
    \lVert f - g \rVert_{C^1}
    \leq \lVert f - \hat{g} \rVert_{C^1} + \lVert \hat{g} - g \rVert_{C^1}  < \frac{\epsilon}{2} + \frac{\epsilon}{2} = \epsilon, \nn
\ea 
which completes the proof.
\end{proof}

\subsubsection{The ESN Approximation Theorem}

In this subsection we will state and prove the ESN Approximation Theorem which states that there exists a linear readout layer $W^{\text{out}}$ giving rise to an autonomous ESN phase with a normally hyperbolic attracting $m$-submanifold on which the autonomous dynamics are topologically conjugate to a structurally stable $\phi$. The idea behind the theorem is observe that the ESN looks enough like a single layer neural network that the Random Universal Approximation Theorem holds. Consequently we can choose linear readout weights stored in the matrix $W^{\text{out}}$ to approximate any $C^1$ function. We will assume that $f$ is an embedding, and therefore invertible on its image, and choose readout weights $W^{\text{out}}$ such that the autonomous ESN approximates a $C^1$ dynamical system possessing an $m$ dimensional normally hyperbolic attracting submanifold on which the dynamics approximate $f \circ \phi \circ f^{-1}$. We want the manifold to be normally hyperbolic and attracting to ensure that an autonomous trajectory that leaves the manifold by some small distance is attracted back toward the manifold, preventing an accumulation of errors from sending the trajectory too far away. Autonomous trajectories originating near the manifold therefore remain near, all the while approximating $f \circ \phi \circ f^{-1}$. To formalise these ideas, we will first define a normally hyperbolic attracting submanifold.
\begin{defn}
    (Normally Hyperbolic Attracting Submanifold) Let $\phi \in \text{Diff}^1(M)$, then, a $\phi$-invariant submanifold $\Lambda \subset M$ is a normally hyperbolic attracting submanifold if the restriction to $\Lambda$ of the tangent bundle of $M$ admits a splitting into a direct sum of two $D\phi$-invariant subbundles, the tangent bundle of $\Lambda$, and the stable bundle $E^s$. Furthermore, with respect to some Riemannian metric on $M$, the restriction of $D\phi$ to $E^s$ must be a contraction, and must be relatively neutral on $T\Lambda$. Thus, there exist constants $0 < \lambda < \mu^{-1} < 1$ and $c > 0$ such that 
    \begin{align}
        T_{\Lambda}M &= T\Lambda \oplus E^s \nn \\
        (D\phi)_xE^s_x &= E^s_{\phi(x)} \ \forall x \in \Lambda \nn \\
        \lVert D\phi^k v \rVert &\leq c\lambda^k \lVert v \rVert \ \forall v \in E^s, \ \forall k \in \mathbb{N} \nn  \\
        \lVert D \phi^k v \rVert &\leq c \mu^{\lvert k \rvert} \lVert v \rVert. \nn
    \end{align}
\end{defn}

Before we present the ESN Approximation Theorem itself we will prove that there exists a $C^1$ evolution operator $\eta$ defined on $\mathbb{R}^d$ that has a normally hyperbolic attracting submanifold on which the dynamics of $\eta$ are conjugate to $\phi$. The existence of this map $\eta$ is guaranteed by standard topological machinery which we recall briefly here, and which is presented in detail by \cite{WarnerManifolds}.

\begin{defn}
    (Cubic centred chart) A chart $(V,\varphi)$ belonging to a $d$-manifold is called a cubic chart if $\varphi(V)$ is an open cube centred about the origin in $\mathbb{R}^d$. If $x \in V$ and $\varphi(x) = 0$, then the chart $(V,\varphi)$ is centred at $x$.
\end{defn}

\begin{defn}
    (Slice coordinates) Suppose that $(V,\varphi)$ is a chart on a $d$-manifold $D$ with coordinate functions $x_1 , ... , x_d$ and that $m$ is an integer $0 \leq m \leq d$. Let $a \in \varphi(V)$ and let
    \begin{align}
        S = \{ q \in V \mid x_i(q) = a_i , i = m+1 , ... , d \}. \nn
    \end{align}
    The subspace $S$ of $D$ together with coordinate maps $x\rvert_{S}$ for $j = 1 , ... , m$
    forms a submanifold of $D$, called a slice of the chart $(V,\varphi)$.
\end{defn}

\begin{lemma}
    (Slice Lemma) Let $M$ be a compact $m$-manifold, let $f : M \to \mathbb{R}^d$ be an immersion, and let $x \in M$. Then there exists a cubic centred chart $(V,\varphi)$ about $f(x)$ and a neighbourhood $U$ of $x$ such that $f\rvert_{U}$ is injective and $f(U)$ is a slice of $(V,\varphi)$. 
\end{lemma}

\begin{proof}
    \cite{WarnerManifolds} page 28 prop 1.35.
\end{proof}

\begin{lemma}
    Let $d > m$ and $M$ be a compact $m$-manifold. Let $\phi \in \text{Diff}^1(M)$. Suppose $f \in C^1(M,\mathbb{R}^d)$ is a $C^1$ embedding. Then there is an open subset $\Omega \in \mathbb{R}^d$ and $\eta \in \text{Diff}^1(\Omega)$ with $f(M)$ a normally hyperbolic attracting submanifold such that $\eta\lvert_{f(M)} = f \circ \phi \circ f^{-1}$ (where we have defined $f^{-1}$ on the image of $f$).
    \label{manifold_lemma}
\end{lemma}

\begin{proof}
    We will make a similar argument to \cite{WarnerManifolds} in the proof of his Proposition 1.36, on page 29. First let $x \in M$. Then by the Slice Lemma there exists a cubic centred chart $(V_x,\varphi_x)$ about $f(x)$ and a neighbourhood $U_x$ of $x$ such that $f(U_x)$ is a slice $(V_x,\varphi_x)$.    
    Let $x_1, \ldots, x_m$ be the slice coordinates in the chart $(V_x,\varphi_x)$ of points in $f(U_x)$.
    Then we can define a map $\eta_x \in \text{Diff}^1(V_x,\mathbb{R}^d)$ applying the map $f \circ \phi \circ f^{-1}$ on the slice co-ordinates and dividing the remaining co-ordinates by 2.   
    We can make this argument for every $x \in M$ hence define a collection of maps $\{ \eta_x \}$ over a collection of open sets $\{ V_x \}$ which cover $f(M)$. Now we let $\{ \alpha_j \ | \ j \in \mathbb{N} \}$ form a partition of unity subordinate to the cover $\{ V_x \}$. We take a subsequence $\{ \alpha_k \}$ such that $\text{supp}(\alpha_k) \cap f(M) \neq \emptyset$ and denote the collection of sets to which $\{ \alpha_k \}$ is subordinate
    by $\{ V_k \}$. We then define a map $\eta$ on a neighbourhood $\Omega : = \cup_{k} V_k$ of $f(M)$ by
    \begin{align}
        \eta = \sum_k \alpha_k \eta_x. \nn
    \end{align}
    By construction, $\eta\rvert_{f(M)} = f \circ \phi \circ f^{-1}$ and $\eta$ has a normally hyperbolic attracting submanifold $f(M)$.
    \end{proof}

Not only does the dynamical system $\eta$ exist, but importantly, its normally hyperbolic attracting submanifold is preserved by any sufficiently good approximation. This is made formal in the Invariant Manifold Theorem, which we will use in the proof of the ESN Approximation Theorem.

\begin{theorem}
    (Invariant Manifold Theorem) Let $K$ be a compact manifold and $\eta \in \text{Diff}^1(K)$ with normally hyperbolic attracting submanifold $\Lambda$. Then, $\exists \epsilon > 0$ such that for any $u \in \text{Diff}^1(K)$ with $\lVert \eta - u \rVert_{C^1} < \epsilon$, the diffeomorphism $u$ has a normally hyperbolic attracting submanifold $U$ such that $\lVert U - \Lambda \rVert_{C^1} < \epsilon$.
\end{theorem}

\begin{proof}
    \cite{Invariant_Manifolds}.
\end{proof}

With these preliminaries established we are ready to prove our ESN Approximation Theorem.
Our strategy involves imposing a special structure on the reservoir matrix $A$ in order to obtain sufficiently many neurons for the Random Universal Approximation Theorem to hold while controlling the dimension of the codomain of the Echo State Map. The structure of $A$ is made clear in the statement of the ESN Approximation Theorem and illustrated in
Figure~\ref{modified_autonomous_ESN}, where we call the connections represented by the matrix $A$ `strongly recurrent' and those represented by $X$ `weakly recurrent'.
The weakly recurrent neurons and the vector $Y$ of inputs are introduced in the proof of
the ESN Approximation Theorem in order to satisfy the conditions of the Random Universal Approximation Theorem. 

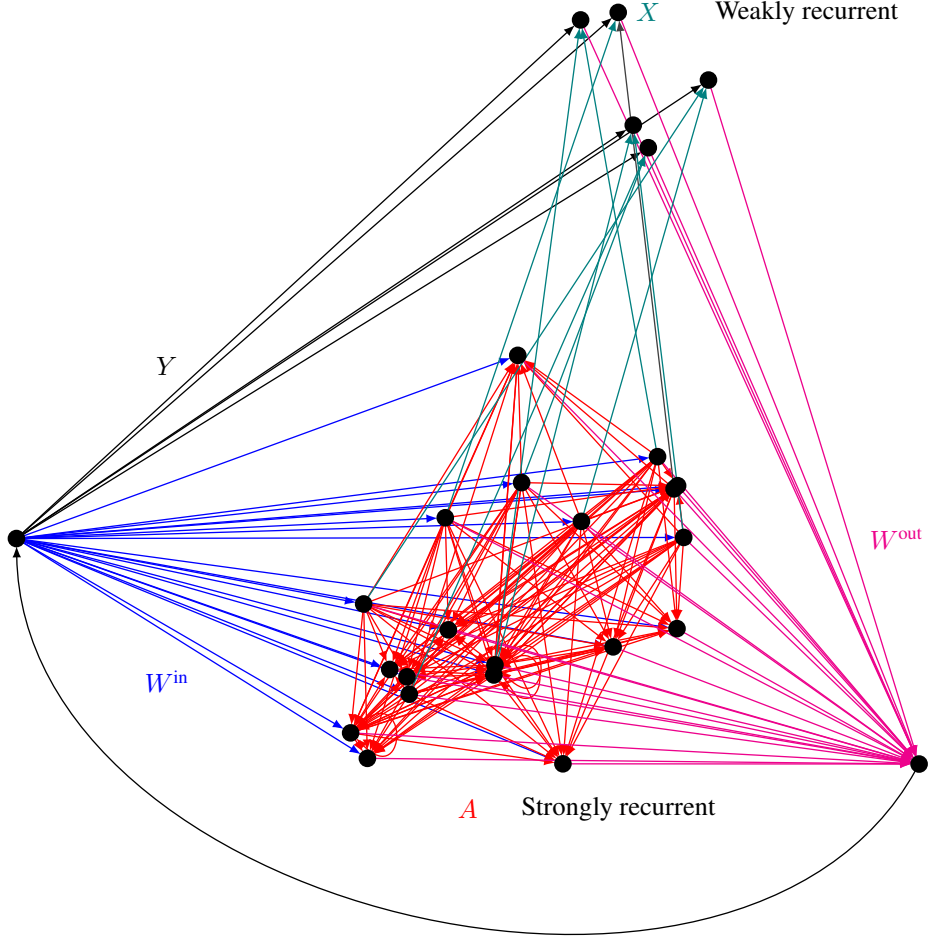
\begin{figure}
\centering
\begin{tikzpicture}
\pgfmathsetseed{1} 
\Vertex[x=-6, y=0, style={color=black},size=.1]{0}
\foreach \x in {1,2,...,20}{
\Vertex[x=rand*3,y=rand*3,style={color=black},size=.1]{\x}}
\foreach \x in {1,2,...,20}{
\Edge[color = blue,Direct,lw=.5pt](0)(\x)}
\foreach \x in {1,2,...,11}{
\foreach \y in {9,10,...,20}{
\Edge[color = red,Direct,lw=.5pt](\x)(\y)}
}
\Vertex[x=6, y=-3, style={color=black},size=.1]{21}
\foreach \x in {1,2,...,20}{
\Edge[color = magenta,Direct,lw=.5pt](\x)(21)}
\Edge[color = black,Direct,lw=.5pt,bend=75](21)(0)
\node[text=blue](Win) at (-4,-1.9){$W^{\text{in}}$};
\node[text=black](Y) at (-4,2.3){$Y$};
\node[text=teal](X) at (2.4,7.0){$X$};
\node[text=red](A) at (0,-3.6){$A$};
\node(SR) at (2,-3.6){Strongly recurrent};
\node(WR) at (4.5,7){Weakly recurrent};
\node[text=magenta](Wout) at (5.7,0){$W^{\text{out}}$};
\Vertex[x=2,y=7,style={color=black},size=.1]{22}
\Vertex[x=2.4,y=5.2,style={color=black},size=.1]{23}
\Vertex[x=3.2,y=6.1,style={color=black},size=.1]{24}
\Vertex[x=1.5,y=6.9,style={color=black},size=.1]{25}
\Vertex[x=2.2,y=5.5,style={color=black},size=.1]{26}
\foreach \x in {22,23,...,26}{
\Edge[color = black,Direct,lw=.5pt](0)(\x)}
\foreach \x in {22,23,...,26}{
\Edge[color = magenta,Direct,lw=.5pt](\x)(21)}
\Edge[color=,Direct,lw=.5pt](1)(22)
\Edge[color=teal,Direct,lw=.5pt](2)(23)
\Edge[color=teal,Direct,lw=.5pt](3)(24)
\Edge[color=teal,Direct,lw=.5pt](4)(25)
\Edge[color=teal,Direct,lw=.5pt](5)(26)
\Edge[color=teal,Direct,lw=.5pt](6)(22)
\Edge[color=teal,Direct,lw=.5pt](7)(23)
\Edge[color=teal,Direct,lw=.5pt](8)(24)
\Edge[color=teal,Direct,lw=.5pt](9)(25)
\Edge[color=teal,Direct,lw=.5pt](10)(26)
\end{tikzpicture}
\caption{The ESN with sparsity structure imposed on $A$ so that we can prove the ESN Approximation Theorem. The matrix $X$ and vector $Y$ are defined in the statement of the ESN Approximation Theorem.}
\label{modified_autonomous_ESN}
\end{figure}

\begin{defn}
    (ESN autonomous phase) The ESN autonomous phase with parameters $(A,W^{\text{in}},W^{\text{out}},\varphi)$ is a discrete time autonomous
    dynamical system $\psi \in C^1(\mathbb{R}^n)$ defined by
    \begin{align}
        \psi(s) = \varphi\big( (A + W^{\text{in}}W^{\text{out}} )s \big). \nn
    \end{align}
\end{defn}

\begin{theorem}
    (ESN Approximation Theorem) Let $M$ be a compact $m$-manifold and $n \in \mathbb{N}$ such that $n > 2m$. Let $A$ be an $n \times n$ matrix where
    $\lVert A \rVert_2 < \text{min}(1 / \lVert D\phi^{-1} \rVert_{\infty}, 1)$, and $W^{\text{in}}$ an $n \times 1$ matrix, and let the triple $(\varphi,A,W^{\text{in}})$ be an ESN. Let $\phi \in \text{Diff}^1(M)$ be structurally stable, and let $\omega \in C^1(M,\mathbb{R})$. Suppose the Echo State Map $f \in C^1( M ,\mathbb{R}^{n})$ is a $C^1$ embedding. Let $(x_j)_{j \in \mathbb{N}}$, $(y_j)_{j \in \mathbb{N}}$, and $(b_j)_{j \in \mathbb{N}}$ be sequences of i.i.d. $\mathbb{R}^n$, $\mathbb{R}$, and $\mathbb{R}$-valued random variables, respectively, with full support. Let $\alpha \in (0,1)$. Then, with probability $\alpha$, there exists $d \in \mathbb{N}$ with $d > n$, 
    a $d \times 1$ matrix $W^{\text{out}}$, a $d \times d$ matrix $\tilde{A}$, and a $d \times 1$ matrix
    $\tilde{W^{\text{in}}}$ assembled from the $n \times n$ matrix $A$, the $(d-n) \times n$ matrix $X$ with $j$th row $x_j$, and the $(d-n) \times 1$ matrix $Y$ with $j$th row $y_j$, like so:
        \begin{align}
        \tilde{A} = 
        \begin{bmatrix}
    A & 0 \\
    X & 0 
        \end{bmatrix}
\qquad \text{ and } \qquad
        \tilde{W}^{\text{in}} =
\begin{bmatrix}
    W^{\text{in}} \\
    Y
\end{bmatrix} \nn,
    \end{align}
    and an activation function
    \begin{align}
        \tilde{\varphi}_i(r) = \sigma(r_i + b_i) \quad  \forall \ i \in \{1 , ... , d \} \nn
    \end{align}
    such that the autonomous ESN $\psi \in C^1(\mathbb{R}^d)$ with parameters $(\tilde{A},\tilde{W^{\text{in}}},W^{\text{out}},\tilde{\varphi})$
    has a normally hyperbolic attracting submanifold on which $\psi$ is topologically conjugate to $\phi$.
    \label{ESN_approximation_theorem}
\end{theorem}

\begin{proof}
    By assumption, the Echo State Map $f$ defined for the ESN $(\varphi,A,W^{\text{in}})$ with respect to $(\phi,\omega)$ is an embedding, so the Echo State Map $\tilde{f}$ defined for $(\tilde{\varphi},\tilde{A},\tilde{W^{\text{in}}})$ with respect to $(\phi,\omega)$ is also an embedding. For the remainder of the proof we will restrict the codomain of $\tilde{f}$ to its image in order to yield a $C^1$ diffeomorphism.
    Before we proceed, we will establish some preliminary results. First we define
     $y : M \to y(M) \subset \mathbb{R}^{n+1}$ by
     \begin{align}
        y_1(x) = \omega(x)
   \qquad \text{ and } \qquad 
            \begin{bmatrix}
            y_2(x) \\
            y_3(x) \\
            \vdots \\
            y_{n+1}(x)
        \end{bmatrix}
        = f \circ \phi^{-1}(x) \nn.
    \end{align}
    Furthermore we will define maps 
    \begin{align}
        \mathcal{F}:C^1(M,\mathbb{R}^d) \to C^1(\tilde{f}(M),\mathbb{R}^d)
    \qquad \text{ by } \qquad 
    \mathcal{F}(g) = g \circ \tilde{f}^{-1} \nn
    \end{align}
    and
    \begin{align}
    \mathcal{Y}:C^1(y(M),\mathbb{R}) \to C^1(M,\mathbb{R}) 
    \qquad \text{ by } \qquad 
    \mathcal{Y}(g) = g \circ y. \nn
    \end{align}
    Next we will show that $\mathcal{F}$ and $\mathcal{Y}$ are Lipschitz continuous. To see that $\mathcal{F}$ is Lipschitz continuous observe
    \ba
        \lVert \mathcal{F}(g) - \mathcal{F}(h) \rVert_{C^1} & = &
        \lVert g \circ \tilde{f}^{-1} - h \circ \tilde{f}^{-1} \rVert_{C^1} \nn \\
        & = & \lVert g \circ \tilde{f}^{-1} - h \circ \tilde{f}^{-1} \rVert_{\infty} + \lVert Dg \circ \tilde{f}^{-1}D\tilde{f}^{-1} - Dh \circ \tilde{f}^{-1} D\tilde{f}^{-1} \rVert_{\infty} \nn \\
        & \leq & \lVert g \circ \tilde{f}^{-1} - h \circ \tilde{f}^{-1} \rVert_{\infty} + \lVert D\tilde{f}^{-1} \rVert_{\infty} \lVert Dg \circ \tilde{f}^{-1} - Dh \circ \tilde{f}^{-1} \rVert_{\infty} \nn \\
        & = & \lVert g - h \rVert_{\infty} + \lVert D\tilde{f}^{-1} \rVert_{\infty}\lVert Dg - Dh \rVert_{\infty} \nn \\
        & \leq & \text{max}(1,\lVert D\tilde{f}^{-1} \rVert_{\infty})( \lVert g - h \rVert_{\infty} + \lVert Dg - Dh \rVert_{\infty} ) \nn \\
        & = & \text{max}(1,\lVert D\tilde{f}^{-1} \rVert_{\infty}) \lVert g - h \rVert_{C^1} \nn.
    \ea
    We can make an almost identical argument to show that $\mathcal{Y}$ is Lipschitz continuous. We will denote the Lipschitz constants for $\mathcal{F}$ and $\mathcal{Y}$ by $L$ and $M$ respectively. We are now ready to proceed with the proof.
    
     By Lemma \ref{manifold_lemma}, there exists an open subset $\Omega \in \mathbb{R}^d$ containing $\tilde{f}(M)$ and $\eta \in \text{Diff}^1(\Omega)$ with $\tilde{f}(M)$ a normally hyperbolic attracting submanifold such that
    \begin{align}
        \eta\rvert_{\tilde{f}(M)} = \tilde{f} \circ \phi \circ \tilde{f}^{-1}. \nn 
    \end{align}
    Now let $K \subset \Omega$ be a compact manifold containing $\tilde{f}(M)$.
   Normally hyperbolic invariant submanifolds persist under small perturbations, by the Invariant Manifold Theorem, so $\exists \ \epsilon > 0$ such that any $u \in \text{Diff}^1(K)$ which satisfies
        $\lVert u - \eta\rvert_K \rVert_{C^1} < \epsilon$
    is topologically conjugate to $\eta$. For any given value $\alpha \in (0,1)$, by the Random Universal Approximation Theorem, there exists a $d \in \mathbb{N}$ and a $d \times 1$ matrix $W^{\text{out}}$ such that $g \in C^1(\mathbb{R}^{n+1},\mathbb{R})$ defined by
    \begin{align}
        g(z) = \sum_{i = 1}^d W^{\text{out}}_i \sigma\bigg(  
        \begin{bmatrix}
            \tilde{W}^{\text{in}} & \tilde{A} 
        \end{bmatrix}
        _i z
        + b_i \bigg) \label{defn_g}
    \end{align}
     satisfies
     \begin{align}
         \rVert g - \omega \circ \phi \circ y^{-1} \lVert_{C^1} < \frac{\epsilon}{LM\lVert W^{\text{in}} \rVert} \label{g<eps}
     \end{align}
     where
     $[\tilde{W}^{\text{in}},\tilde{A}]_i$ is a $1 \times (n+1)$ matrix with 1st entry $\tilde{W}^{\text{in}}_i$ and $(j+1)$th entry $\tilde{A}_{ij}$.
     Now
    \ba
            \lVert \psi\rvert_{\tilde{f}(M)} - \eta\rvert_{\tilde{f}(M)} \rVert_{C^1}
            & \leq & L\lVert \psi \circ\tilde{f} - \eta \circ \tilde{f} \rVert_{C^1} \nn
            \\ 
            & = & L \lVert \psi \circ \tilde{f} - \tilde{f} \circ \phi \rVert_{C^1} \text{ because $\eta\rvert_{\tilde{f}(M)} = \tilde{f} \circ \phi \circ \tilde{f}^{-1}$} \nn  \\
            & = & L \lVert \tilde{\varphi}(\tilde{A}\tilde{f} + \tilde{W}^{\text{in}}W^{\text{out}}\tilde{f}) - \tilde{\varphi}(\tilde{A}\tilde{f} + \tilde{W}^{\text{in}}\omega\circ\phi) \rVert_{C^1} \text{ by definition of $\psi$} \nn  \\
            & \leq & L \lVert (\tilde{A}\tilde{f} + \tilde{W}^{\text{in}}W^{\text{out}}\tilde{f}) - (\tilde{A}\tilde{f} + \tilde{W}^{\text{in}}\omega\circ\phi) \rVert_{C^1} \text{ because $\tilde{\varphi}$ is contracting} \nn \\
            & = & L \lVert ( \tilde{W}^{\text{in}}W^{\text{out}}\tilde{f}) - (\tilde{W}^{\text{in}}\omega\circ\phi) \rVert_{C^1} \text{ because $\tilde{A}\tilde{f} - \tilde{A}\tilde{f} = 0$} \nn \\
            & \leq & L \lVert \tilde{W}^{\text{in}} \rVert_2 \lVert W^{\text{out}}\tilde{f} - \omega\circ\phi) \rVert_{C^1} \text{ by factoring out $W^{\text{in}}$}  \nn \\
             & = & L \lVert \tilde{W}^{\text{in}} \rVert_2 \lVert W^{\text{out}}\tilde{\varphi}(\tilde{A}\tilde{f} \circ \phi^{-1} + \tilde{W}^{\text{in}} \omega) - \omega \circ \phi \rVert_{C^1} \text{ by Theorem~\ref{recursion_relation}} \nn \\
            & = & L \lVert \tilde{W}^{\text{in}} \rVert_2 \bigg\lVert \sum_{i = 1}^d W^{\text{out}}_i \sigma\bigg(  
        \begin{bmatrix} \tilde{W}^{\text{in}} & \tilde{A} 
        \end{bmatrix}
        _i y
        + b_i \bigg) - \omega \circ \phi \bigg\rVert_{C^1} \text{ by definition of $\tilde{\varphi}$ and $y$} \nn \\
        & = & L \lVert \tilde{W}^{\text{in}} \rVert_2 \lVert g \circ y - \omega \circ \phi \rVert_{C^1} \text{ by \eqref{defn_g}} \nn \\
        & \leq & LM \lVert \tilde{W}^{\text{in}} \rVert_2 \lVert g\lvert_{y(M)} - \omega \circ \phi \circ y^{-1}
    \rVert_{C^1} \nn \\
    & < & LM \lVert \tilde{W}^{\text{in}} \rVert_2 \frac{\epsilon}{LM\lVert \tilde{W}^{\text{in}} \rVert_2} \nn \text{ by \eqref{g<eps}} \\
    & = & \epsilon \nn
    \ea
hence there is some open set $\tilde{\Omega} \subset K$ containing $\tilde{f}(M)$ such that
    \begin{align}
        \lVert \psi\lvert_{\tilde\Omega} - \eta\lvert_{\tilde\Omega} \rVert_{C^1} < \epsilon \nn 
    \end{align}
    so $\psi\lvert_{\tilde\Omega}$ is conjugate to $\eta\lvert_{\tilde\Omega}$. Consequently, there exists an $h \in \text{Diff}^1(\tilde\Omega)$ such that $\psi\lvert_{\tilde\Omega} = h \circ \eta\lvert_{\tilde\Omega} \circ h^{-1}$. Now $\tilde{f}(M)$ is a normally hyperbolic attracting submanifold of $\eta$ where $\eta\lvert_{\tilde{f}(M)} = \tilde{f} \circ \phi \circ \tilde{f}^{-1}$ so $h \circ \tilde{f}(M)$ is a normally hyperbolic attracting submanifold of $\psi$ on which
    \begin{align}
        \psi = h \circ \eta \circ h^{-1} = h \circ \tilde{f} \circ \phi \circ \tilde{f}^{-1} \circ h^{-1} \cong \phi. \nn 
    \end{align}
\end{proof}
\begin{remark}
    A consequence of the ESN Approximation Theorem is that the diagram shown in Figure~\ref{Commuting_diagram} commutes.
    \begin{figure}
  \centering
        \begin{tikzcd}
        M \arrow{d}{\phi} \arrow{r}{\tilde{f}} & \tilde{f}(M) \arrow{d}{\eta} \arrow{r}{h} & h \circ \tilde{f}(M) \arrow{d}{\psi} \\
        M \arrow{r}{\tilde{f}} & \tilde{f}(M) \arrow{r}{h} & h \circ \tilde{f}(M)
\end{tikzcd}
    \caption{A commuting diagram representing the ESN Approximation Theorem where the terms are defined throughout the theorem's proof.}
    \label{Commuting_diagram}
\end{figure}
\end{remark}
\section{Numerical Experiments with ESNs}
\label{sec::numerical}

In the previous section we showed that for a given structurally stable dynamical system and a sufficiently large ESN there exists a linear output matrix $W^{\text{out}}$ that gives rise to an autonomous ESN with dynamics that are topologically conjugate to those of the given dynamical system.

To test whether these results hold in practice we took a 1D observation of a numerically integrated trajectory of the Lorenz system, fed this into an ESN implemented on a commercial laptop, and sought to discover whether the autonomous phase of the ESN would adopt dynamics topologically conjugate to the Lorenz system. In particular we computed several topological invariants of the ESN autonomous phase including the Lyapunov exponents, fixed point eigenvalues, and homology, then compared these to the known invariants of the Lorenz system. This work was inspired by a paper by \cite{Pathak2017} who trained an ESN on a full 3D trajectory of the Lorenz system, rather than a 1D observation, and compared the Lyapunov exponents of the autonomous phase to the known exponents of the Lorenz system. In a more recent works, \cite{Vlac_2019} train an ESN on 1D observations of the Lorenz-96 system, and also compare the Lyapunov exponents of the autonomous phase to the known exponents of the Lorenz system. \cite{Ashesh_2019} also train an ESN on observations of the Lorenz-96 system and evaluate the accuracy of future prediction for reservoirs of different size.

We used MATLAB's ODE45 to integrate a trajectory of the \cite{doi:10.1175/1520-0469(1963)020<0130:DNF>2.0.CO;2} system
\begin{align}
    \dot{x} &= \sigma(y - x) \\
    \dot{y} &= x(\rho - z) - y \nonumber \\
    \dot{z} &= xy - \beta z \nonumber
\end{align}
with parameters $\sigma = 10, \beta = 8/3, \rho = 28$ chosen so the system produces the celebrated Lorenz attractor shown in Figure \ref{Lorenz_attractor_fig}.
\begin{figure}
  \centering
    \includegraphics[width=0.7\textwidth]{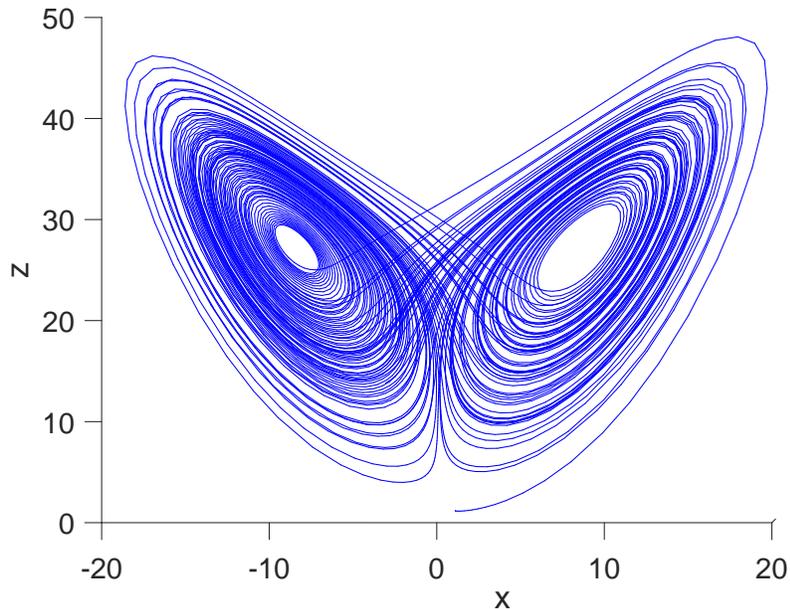}
    \caption{A picture of the famous Lorenz attractor. Here the trajectory was initialised at $(1,1,1)$ and quickly converges to the attractor.}
    \label{Lorenz_attractor_fig}
\end{figure}
We then observed the $x$ component of the trajectory by choosing the observation function $\omega(x,y,z) = x$ to create a 1 dimensional time series. We fed this time series into an ESN with the following parameters: spectral radius $\rho = 1$, reservoir size $n = 300$, and activation function $\varphi = \text{tanh}$. The reservoir matrix $A$ is an Erd\H{o}s-R\'{e}nyi matrix with mean $6$ and connection weights (where they are non-zero) i.i.d Gaussian, re-scaled such that $\rho = 1$. The keen reader will notice that the structure of $A$ does not conform to the reservoir matrix $\tilde{A}$ described in the statement of the ESN Approximation Theorem. The fact that our numerical experiments produce good results despite this suggests this weakly connected $\tilde{A}$ is unnecessary, but rather a decision we made to make the ESN Approximation Theorem easier to prove. Furthermore, insisting that $\rho < 1$ is not sufficient in to ensure that $\lVert A \rVert_2 < 1$, but this is a common choice in practical applications. The matrix $W^{\text{out}}$ is populated with i.i.d Gaussian weights $\sim\mathcal{N}(0,1)$ which 
are then scaled by a `strength parameter' $p = 0.1$. We choose a regularisation parameter $\lambda = 10^{-6}$ to solve the regularised least squares problem 
\begin{align}
    \min_{W^{\text{out}}}\sum_{k=1}^{K} \lVert W^{\text{out}}r_{k} - u_{k} \rVert^2 + \lambda\lVert W^{\text{out}} \rVert^2_2 \nn
\end{align}
using the SVD method
presented by \cite{Deblurring_2006}.
We will note here that the linear output layer $W^{\text{out}}$ obtained by this procedure is not necessarily the same as that guaranteed by the ESN Approximation Theorem.
These parameters were
carefully hand tuned so that the autonomous phase appeared (by eye) to match the driven phase. The question of how to systematically choose good parameters is discussed by \cite{Bayesian_Optimisation_of_ESN} who searched through parameter space using Bayesian optimisation, and used cross validation to test the goodness of fit. Now, 
with $W^{\text{out}}$ obtained, we ran the autonomous ESN and plotted the future observations $v_i$ in Figure \ref{x_v_t_fig}. We can see from this Figure that the ESN seems to predict the qualitative features of the future trajectory very well. 
\begin{figure}
  \centering
    \includegraphics[width=1\textwidth]{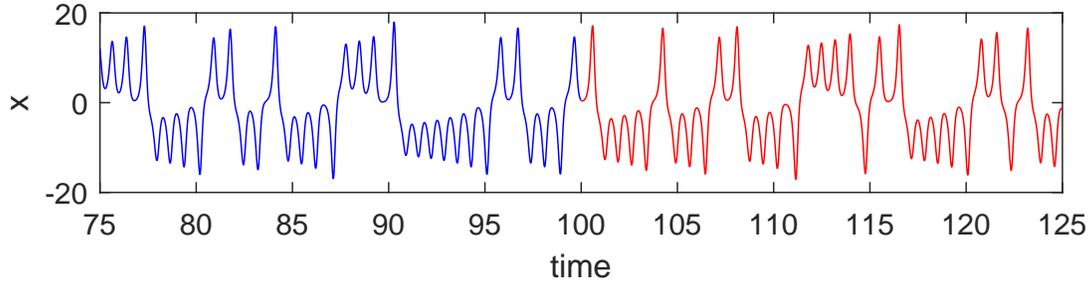}
      \caption{Here the 1D observations are shown in blue (up to time 100 for those of you reading in black and white) and future predictions shown in red (onwards from time 100).}
      \label{x_v_t_fig}
\end{figure}

Since the Lorenz system is defined on a 3-manifold, we can usefully plot trajectories of the entire system. To check by eye whether the reservoir dynamics of both the driven phase and autonomous phase are topologically conjugate to the Lorenz dynamics, we projected the driven and autonomous dynamics onto the first 3 principal components of the driven trajectory and present them in Figure \ref{fixed_point_full.fig}.

\subsection{Locating Fixed Points and Determining their Eigenvalues}

If the ESM $f$ is an embedding, then $f$ will embed the fixed points of the Lorenz system into the reservoir space. Moreover if the autonomous ESN approximates the embedded Lorenz system on a neighbourhood of the embedded fixed points sufficently well, the autonomous dynamics will contain fixed points very close to those of the embedded Lorenz system. To verify this, we searched for the autonomous ESN's fixed points using Newton's method, and found them, as
illustrated in Figure~\ref{fixed_point_full.fig}.

\begin{figure}
  \centering
    \includegraphics[width=0.7\textwidth]{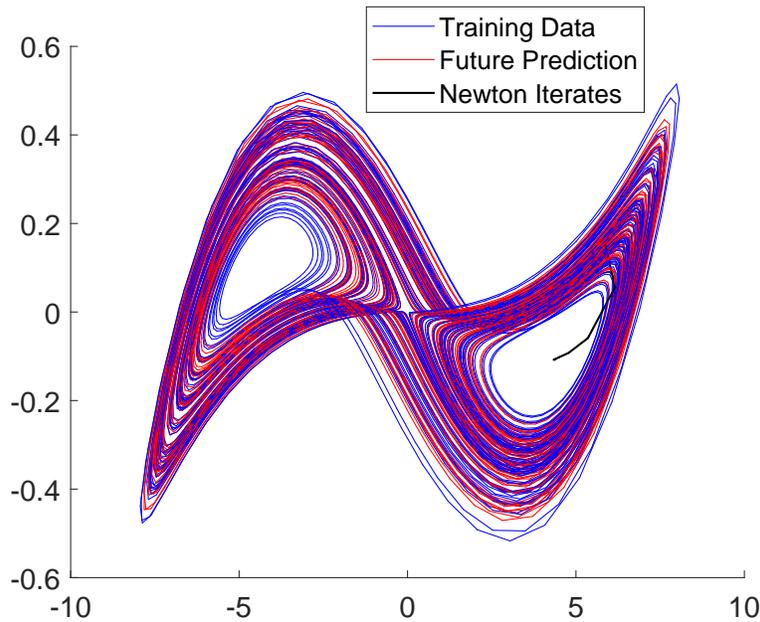}
    \caption{The driven reservoir dynamics are plotted in blue and autonomous dynamics are plotted in red. Both were projected onto the first three principal components of the driven dynamics, then the axes are rotated such that the projection appears on the first 2 components. The black line indicates the iterates of Newton's method, used to locate a fixed point - the method eventually converges to a fixed point in the middle of the right wing of the figure. We can see by eye that the reservoir dynamics appear by eye to be topologically conjugate to the Lorenz system.}
    \label{fixed_point_full.fig}
\end{figure}

Further, if the ESM $f$ is a $C^1$ embedding of the original dynamics, we expect $f$ to preserve the stability of fixed points, i.e. we expect the eigenvalues of the linearisation of the autonomous phase to be preserved at every fixed point. Now, comparing the eigenvalues of the linearisation of the Lorenz system and autonomous phase at the respective fixed points requires some subtlety, because the Lorenz system is a continuous time flow, while the autonomous phase is a discrete time map. So, we began by considering one of the known fixed points found in the Lorenz attractor's wings
\begin{align}
    x^* = (\sqrt{\beta(\rho - 1)},\sqrt{\beta(\rho-1)},\rho - 1), \nn
\end{align}
and noted the Jacobian $J$ of the continuous time Lorenz system evaluated at the fixed point $x^*$ is therefore
\begin{align}
J\Big|_{x^*} =
\begin{bmatrix}
    -\sigma & \sigma & 0 \\
    1 & -1 & -\sqrt{\beta(\rho - 1)} \\
    \sqrt{\beta(\rho-1)} & \sqrt{\beta(\rho-1)} & -\beta
\end{bmatrix}. \nn
\end{align}
Now we can discretise the Lorenz system $\dot{x} = s(x)$ with the following map
\begin{align}
    x_{k+1} = x_k + \int_{t_k}^{t_{k+1}} s \circ x(t) dt,
    \nn 
\end{align}
hence the discrete time linearisation about the fixed point $x^*$ is
\begin{align}
    x_{k+1} = \exp\bigg(J\Big|_{x^*}(t_{k+1}-t_k)\bigg)x_k, \nn
\end{align}
which has 3 eigenvalues, which we have compared with the ESN autonomous eigenvalues in Figure \ref{eigenvalues-fig}.
\begin{figure}
  \centering
    \includegraphics[width=0.7\textwidth]{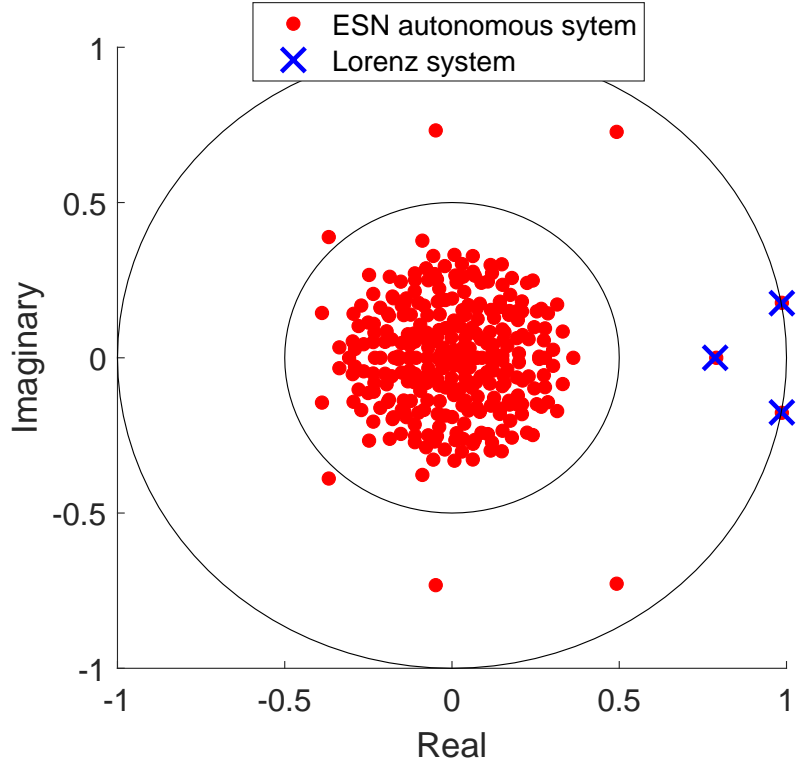}
      \caption{Here the 3 eigenvalues of the linearisation of the Lorenz system on the fixed point inside one of the Lorenz attractor's wings are represented by blue crosses. The 300 eigenvalues of the linearisation of the ESN autonomous system at the fixed point found with Newton's method are represented by red dots.}
      \label{eigenvalues-fig}
\end{figure}
If the ESM $f$ is indeed a $C^1$ embedding, the dynamics of the autonomous phase are topologically conjugate to the discrete time Lorenz system on some 3-submanifold. This manifold is spanned by 3 eigenvectors, each with an associated eigenvalue, which will coincide with the eigenvalues of the linearisation of the Lorenz system on the fixed point. Figure \ref{eigenvalues-fig} appears to show 3 overlapping eigenvalues, suggesting that the autonomous phase is diffeomorphic to the Lorenz system (at least in a neighbourhood of $x^*$) in this simulation. This is particularly remarkable because $x^*$ is distant from the training data. The ESN has successfully inferred the existence, position and eigenvalues of a fixed point from training data, which contains no fixed points. In the machine learning parlance, the ESN has generalised patterns in the training data to an unseen region of the phase space.

\subsection{Comparison of Lyapunov Spectra}

Another topological invariant of the Lorenz system is the Lyapunov spectrum, which captures how quickly very close trajectories diverge from eachother, and is used as a measure of chaos. To define the spectrum, let $J$ be the Jacobian of the evolution operator of a continuous time dynamical system. Let $Y$ be the solution of the ODE $\dot Y = JY$ with initial condition $Y(0) = x_0$. Then the Lyapunov Spectrum of the invariant set containing $x_0$ is the spectrum of the matrix $\Lambda$ defined
    \begin{align}
        \Lambda = \lim_{t \to \infty} \frac{1}{2t} Y Y^{\top}. \nn
    \end{align}
 Each eigenvalue in the spectrum is called a \emph{Lyapunov exponent} to signify that two initially close trajectories diverge or converge exponentially fast with exponentiation constant in the direction of each eigenvector of $J$ given by a Lyapunov exponent. Details are discussed by \cite{DARBYSHIRE1996287}. The Lyapunov spectrum for the Lorenz system was estimated by \cite{Sprott_2003} as 0.9056, 0, -14.5723. In order to compare the Lorenz spectrum to the spectrum of the autonomous ESN, we computed the autonomous system's spectrum using the discrete time $QR$ method discussed in \cite{DARBYSHIRE1996287} and plotted each Lyapunov exponent against the known exponents of the Lorenz system in Figure \ref{Lorenz_Lyapunov}. We found the largest 2 in good agreement while there was significant error in the smallest, which is a common problem also encountered by \cite{Pathak2017}.
\begin{figure}
  \centering
    \includegraphics[width=0.7\textwidth]{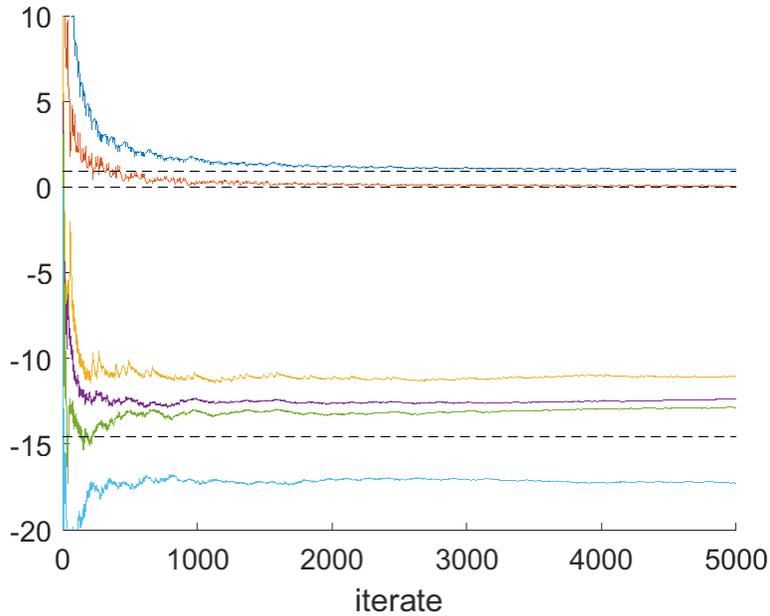}
    \caption{The Lyapunov spectrum of the autonomous phase as the iterates increases is shown. The true Lyapunov exponents of the autonomous phase is given by the limit of these exponents as the iterations tend to infinity. These autonomous exponents are compared to the black dotted lines representing the 3 exponents of the Lorenz system.}
    \label{Lorenz_Lyapunov}
\end{figure}

\subsection{Persistent Homology}

We compared the homology groups of the Lorenz attractor to the persistent homology groups of the autonomous and driven attractors.
We followed the lead of \cite{GARLAND201649} who computed the persistent homology of the Lorenz system reconstructed from a sequence of 1D observations of a Lorenz trajectory using the delay observation map described in Takens' Theorem. The authors used the open source software Javaplex created by \cite{Javaplex} to find the Witness Complex for the delay embedded Lorenz attractor and computed the homology of the complex. They discuss a few subtleties that arise, in particular that the Lorenz attractor is a fractal, whose structure cannot be reconstructed exactly from any finite number of sample points. The authors therefore satisfied themselves by approximating the Lorenz attractor with a branched manifold model presented by \cite{PMIHES_1979__50__73_0} which has the homology of the figure 8. We made the same approximation, and expected to find that the application of persistent homology to the Lorenz system, driven ESN dynamics, and autonomous ESN dynamics would reveal that all three have the figure 8 homology groups. In particular the persistence diagrams of these three systems would exhibit a pair of $H_1$ persistent homology groups floating well above the diagonal. To verify this, we produced persistence diagrams using the open source software Ripser produced by \cite{ctralie2018ripser} and plotted the results in Figure \ref{Persistence_fig}.

The reader may wonder why we would use persistent homology to show that the Lorenz system, driven ESN dynamics, and autonomous ESN dynamics all have the homology of the figure 8 when this can clearly be seen in Figures \ref{Lorenz_attractor_fig} and \ref{fixed_point_full.fig}.  The homology of a 3D system is usually apparent from a plot, but persistent homology can reveal the holes, voids and higher dimensional hypervoids of high dimensional systems that cannot be easily visualised. For example \cite{MULDOON19931} computed the homology of a delay embedded time series from a fluid dynamics experiment, which could in general be of much higher dimension.

\begin{figure}
  \centering
    \includegraphics[width=0.7\textwidth]{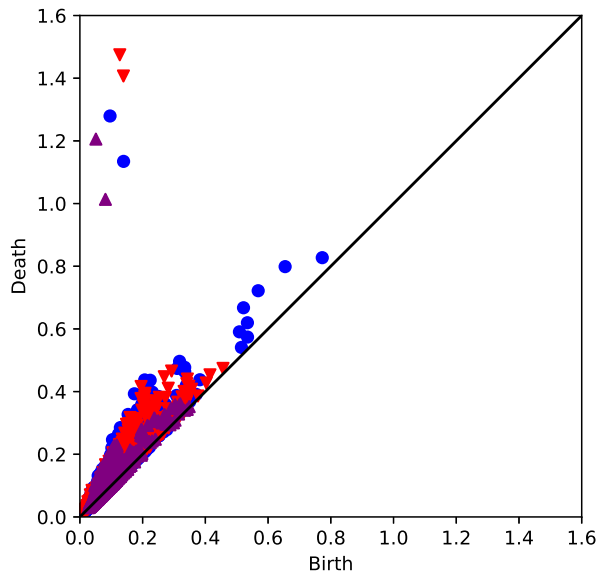}
    \caption{We have plotted the $H_1$ persistence diagrams of the driven ESN dynamics, autonomous ESN dynamics, and Lorenz dynamics as blue circles, red downward triangles, and purple upward triangles. We can see that each of these 3 objects has a pair of points floating well above the diagonal, suggesting each has 2 holes. This is consistent with our expectation that all three adopt the topology of the figure 8.}
    \label{Persistence_fig}
\end{figure}

\section{Conclusions and Outlook}

In this paper, we showed that an Echo State Network driven by a sequence of one dimensional observations of a dynamical system, evolving on a manifold $M$, induces a map $f \in C^1(M,\mathbb{R}^n)$, which we called the Echo State Map. We proved that for a randomly initialed ESN and generic observation function $\omega$, that $f$ is an embedding with positive probability, and called this the weak ESN Embedding Theorem. We conjectured that the theorem holds with probability 1, by analogy to Takens' Theorem. We went on to show that a randomly initialised ESN has a universal approximation property and called this the Random Universal Approximation Theorem (RUAT). Finally, we used both the RUAT and Embedding Theorem to prove that for an ESN trained a sequence of scalar observations of a structurally stable dynamical system, there is a choice of linear readout weights $W^{\text{out}}$ for which the autonomous ESN has dynamics that are topologically conjugate to the input dynamical system, and we called this the ESN Approximation Theorem.

The theory presented here leaves some questions unanswered. 
In practice we use regularized least squares regression to learn an output matrix from the one-dimensional and finite training trajectory, but currently, we have no guarantee that this will result in an autonomous phase ESN that is topologically conjugate to the underlying dynamical system. This is analogous to the case of the Universal Approximation Theorem for feed forwards neural networks, where the theoretical result proves the existence of suitable set of weights but does not guarantee that a particular learning algorithm will be able to find them or how much training data may be required. It may be that imposing extra conditions on the target dynamical system, like ergodicity, allows us to prove that $W^{\text{out}}$ obtained by least squares regression results in an arbitrarily good approximation. This seems to be supported by the experiments in Section \ref{sec::numerical}.

Furthermore, it seems worthwhile to prove the ESN Embedding Conjecture, or some modification of it that is actually correct, by carefully modifying the proof of Takens' Theorem provided by \cite{Hukes_thm}. A sceptical reader may wonder why we would bother using an ESN to embed the trajectory in the first place, when a delay embedding would do. The reason being that it seems the ESN's learning and predictive powers are much more resilient to noise than the simple delay embedding presented by Takens. Heuristically it seems as an observed trajectory passes through the ESN, the noise cancels itself out by taking a nonlinear combination of positive and negative noise. We could therefore view the ESN as a nonlinear filter, generalising the linear filters discussed by \cite{Sauer1991} in the context of \emph{embedology} - the art building delay observation maps with special features, which include being more resiliant to noise than Takens' original map. Understanding the noise cancelling benefits of the ESN could be a fruitful direction of future work. 

Many of the assumptions we made throughout this paper are likely stronger than they need to be. For example \cite{Sauer1991} prove versions of Takens' Theorem for dynamics on a compact invariant set with real box counting dimension - generalising dynamics on a manifold with integer dimension. This is particularly worthwhile because chaotic attractors of interest often lie on invariant sets with non-integer dimension, with the Lorenz attractor serving as a perfect example. We also create a strangely shaped reservoir $\tilde{A}$ in our proof of the ESN Approximation Theorem, which numerical experiments suggests is unnecessary.

\section*{Acknowledgements}
We offer gracious thanks to the anonymous reviewers for their detailed suggestions. In particular, we thank Juan-Pablo Ortega and Lyudmila Grigoryeva for thorough and fruitful discussions at Universit{\"a}t Sankt Gallen which lead to significant improvements to the manuscript. 
\\
\\
We also thank Allen Hart's PhD confirmation examiners Alastair Spence and Chris Guiver who helpfully criticised and improved the content of this paper. 
\\
\\
We acknowledge that Allen Hart is supported by a scholarship from the ESPRC Centre for Doctoral Training in Statistical Applied Mathematics at Bath (SAMBa), under the project EP/L015684/1.

    \bibliographystyle{agsm}
\bibliography{references}

\end{document}